\PassOptionsToPackage{pdfpagelabels=false}{hyperref}
\documentclass[twoside,11pt,final]{entics}
\usepackage{enticsmacro}
\usepackage{graphicx}
\usepackage[all]{xy}
\usepackage{url}

\usepackage{amsmath}
\usepackage{amssymb}
\usepackage{wasysym}
\usepackage{stmaryrd}
\SetSymbolFont{stmry}{bold}{U}{stmry}{m}{n}
\SetSymbolFont{wasy}{bold}{U}{wasy}{m}{n}
\usepackage{multicol}
\usepackage{multirow}
\usepackage{array}
\usepackage{marvosym}
\usepackage[notext,nomath]{stix}
\usepackage{proof}
\usepackage[dvipsnames]{xcolor}
\usepackage{quiver}
\usepackage{tabularx}
\usepackage{soul}
\newcolumntype{Y}{>{\raggedright\arraybackslash\scriptsize}X}

\usepackage[dvipsnames]{xcolor}

\usepackage{etoolbox}
\usepackage[nameinlink,capitalise]{cleveref}

\makeatletter
\patchcmd{\author@fmt}{\mathrm{\scriptsize #2}}{\mathrm{#2}}{}{\PatchFailed}
\makeatother

\crefname{thm}{theorem}{theorems}
\crefname{defn}{definition}{definitions}
\crefname{lem}{lemma}{lemmas}
\crefname{cor}{corollary}{corollaries}
\crefname{prop}{proposition}{propositions}
\crefname{exmp}{example}{examples}
\AddToHook{env/defn/begin}{\crefalias{thm}{defn}}
\AddToHook{env/lem/begin}{\crefalias{thm}{lem}}
\AddToHook{env/cor/begin}{\crefalias{thm}{cor}}
\AddToHook{env/prop/begin}{\crefalias{thm}{prop}}
\AddToHook{env/exmp/begin}{\crefalias{thm}{exmp}}

\def\conf{MFPS 2026} 	
\volume{6}			
\def\volu{6}			
\def\papno{7}			
\def\mydoi{19631}
\def\lastname{Cohen, Grunfeld} 
\def\runtitle{Multicategorical Semantics for Untyped Effects} 
\def\copynames{L. Cohen, A. Grunfeld}    

\newcommand{\interp}[1]{\left\llbracket #1 \right\rrbracket}

\newcommand{\mca}[0]{\mathbb{A}}

\newcommand{\N}[0]{\mathbb{N}}
\newcommand{\lApp}[0]{,}
\newcommand{\lSep}[0]{,}
\newcommand{\lPar}[1]{\left( #1 \right)}
\newcommand{\paren}[1]{\left( #1 \right)}
\newcommand{\lEmpty}[0]{\emptyset}
\newcommand{\lLength}[1]{ \mid\!\! #1\!\! \mid }

\newcommand{\xSign}[0]{\Sigma}
\newcommand{\xCLC}[0]{\lambda_{\text{ml}*}}

\newcommand{\xVAR}[0]{\mathsf{Var}}
\newcommand{\xFuncSMB}[0]{\mathsf{Func}}
\newcommand{\xProcSMB}[0]{\mathsf{Proc}}
\newcommand{\xFuncSymb}[1]{\xFuncSMB\left(#1\right)}
\newcommand{\xProcSymb}[1]{\xProcSMB\left(#1\right)}

\newcommand{\xVALSMB}[0]{\mathbf{VAL}}
\newcommand{\xVAL}[1]{\xVALSMB\left(#1\right)}
\newcommand{\xCMPSMB}[0]{\mathbf{CMP}}
\newcommand{\xCMP}[1]{\xCMPSMB\left(#1\right)}

\newcommand{\TxRetSMB}[0]{\mathsf{return}}
\newcommand{\TxLetSMB}[0]{\mathsf{let}}
\newcommand{\TxRet}[1]{\TxRetSMB\;#1}
\newcommand{\TxLet}[3]{\TxLetSMB \; #1 \Leftarrow #2 \; \mathsf{in} \; #3}
\newcommand{\TxAbs}[2]{\lambda #1 \, . \, #2}
\newcommand{\TxApp}[2]{#1 \, #2}

\newcommand{\JxWFV}[2]{ #1 \vdash_{\textbf{v}} #2 }
\newcommand{\JxWFC}[2]{ #1 \vdash_{\textbf{c}} #2 }
\newcommand{\JxWFT}[2]{ #1 \vdash #2 }

\newcommand{\JxEq}[3]{#1 \vdash #2 \equiv #3}
\newcommand{\JxEqVal}[3]{#1 \vdash_{\textbf{v}} #2 \equiv #3}
\newcommand{\JxEqCmp}[3]{#1 \vdash_{\textbf{c}} #2 \equiv #3}
\newcommand{\JxSat}[5][]{#2 \Vdash_{#1} #4 \equiv_{#3} #5}
\newcommand{\JxValid}[4][]{\Vdash_{#1} #3 \equiv_{#2} #4}
\newcommand{\JxSatVal}[4]{\JxSat[\textbf{v}]{#1}{#2}{#3}{#4}}
\newcommand{\JxValidVal}[3]{\JxValid[\textbf{v}]{#1}{#2}{#3}}
\newcommand{\JxSatCmp}[4]{\JxSat[\textbf{c}]{#1}{#2}{#3}{#4}}
\newcommand{\JxValidCmp}[3]{\JxValid[\textbf{c}]{#1}{#2}{#3}}

\newcommand{\VxTrm}[1][]{E_{#1}}
\newcommand{\VxVal}[1][]{V_{#1}}
\newcommand{\VxCmp}[1][]{M_{#1}}
\newcommand{\VxVar}[1][]{x_{#1}}
\newcommand{\VxVarI}[1][]{x_{#1}}
\newcommand{\VxVarII}[1][]{y_{#1}}
\newcommand{\VxVarIII}[1][]{z_{#1}}
\newcommand{\VxVarIV}[1][]{w_{#1}}
\newcommand{\VxFunc}[1][]{f_{#1}}
\newcommand{\VxProc}[1][]{p_{#1}}
\newcommand{\VxCxt}[1][]{\Gamma_{#1}}

\newcommand{\cApp}[0]{\circledast}
\newcommand{\cAbs}[2][]{\triangleleft^{#1} #2 \triangleright}

\newcommand{\VcCmp}[1][]{f_{#1}}

\newcommand{\VcCmpII}[1][]{g_{#1}}
\newcommand{\VcCmpIII}[1][]{h_{#1}}

\newcommand{\mFreeSMB}[1][]{\boxempty}

\newcommand{\oVal}[1][]{\mathbb{V}_{#1}}
\newcommand{\oCmp}[1][]{\mathbb{C}_{#1}}
\newcommand{\oV}[2][]{\oVal[#1]\left(#2\right)}
\newcommand{\oC}[2][]{\oCmp[#1]\left(#2\right)}

\newcommand{\oFree}[1][]{\mathcal{F}_{#1}}
\newcommand{\oForget}[1][]{\mathcal{U}_{#1}}
\newcommand{\oReturn}[1][]{\mathcal{J}_{#1}}

\newcommand{\oId}[1][]{\mathsf{id}_{#1}}
\newcommand{\oSub}[1]{ \left\{ #1 \right\} }
\newcommand{\oSubAp}[2]{ #1 \oSub{#2} }
\newcommand{\oRen}[1]{ \left[#1\right] }
\newcommand{\oRenAp}[2]{ #1 \oRen{#2} }
\newcommand{\oSL}[0]{\!\dashv}
\newcommand{\oSR}[0]{\vdash\!}
\newcommand{\oHom}[3]{#1\left(#2 \Rightarrow #3\right)}
\newcommand{\oPshApSMB}[0]{\star}
\newcommand{\oPshAp}[2]{#1 \oPshApSMB #2}

\newcommand{\oFreydOp}[0]{\mathbf{FrOp}}
\newcommand{\oFREYDPROP}[0]{\mathbf{FrPROP}}

\newcommand{\oRENSMB}[0]{\mathbf{Ren}}
\newcommand{\oREN}[2]{\oHom{\oRENSMB}{#1}{#2}}
\newcommand{\oPERMSMB}[0]{\mathbf{Perm}}
\newcommand{\oPERM}[1]{\oPERMSMB\left(#1\right)}
\newcommand{\oSUBSMB}[1][]{\mathbf{Sub}_{#1}}
\newcommand{\oSUB}[3][]{\oHom{\oSUBSMB[#1]}{#2}{#3}}
\newcommand{\oPSUBSMB}[1][]{\mathbf{PreSub}_{#1}}
\newcommand{\oPSUB}[3][]{\oHom{\oPSUBSMB[#1]}{#2}{#3}}
\newcommand{\oSSUBSMB}[1][]{\mathbf{Sub}^{\boldsymbol{1}}_{#1}}
\newcommand{\oSSUB}[3][]{\oHom{\oSSUBSMB[#1]}{#2}{#3}}

\newcommand{\oCSUBSMB}[1][]{\mathbf{Sub}^{\times}_{#1}}
\newcommand{\oCSUB}[3][]{\oHom{\oCSUBSMB[#1]}{#2}{#3}}

\newcommand{\oFuncAsgSMB}[0]{ \sigma^{F} }
\newcommand{\oProcAsgSMB}[0]{ \sigma^{P} }
\newcommand{\oFuncAsg}[1]{ \oFuncAsgSMB\left(#1\right) }
\newcommand{\oProcAsg}[1]{ \oProcAsgSMB\left(#1\right) }

\newcommand{\oInterpVal}[2][]{\interp{#2}_{#1}^{V}}
\newcommand{\oInterpCmp}[2][]{\interp{#2}_{#1}^{C}}

\newcommand{\oApp}[0]{ \cApp }
\newcommand{\oAbs}[2][]{ \cAbs[#1]{ #2 } }
\newcommand{\oCurry}[1]{\left\lceil #1 \right\rceil}
\newcommand{\oUncurry}[1]{\left\lfloor #1 \right\rfloor}
\newcommand{\oSym}[1][]{\sigma_{#1}}
\newcommand{\oCopy}[1][]{\Delta_{#1}}

\newcommand{\oDiscard}[1][]{!_{#1}}

\newcommand{\VoVal}[1][]{v_{#1}}
\newcommand{\VoValI}[1][]{u_{#1}}
\newcommand{\VoValII}[1][]{v_{#1}}

\newcommand{\VoCmp}[1][]{f_{#1}}
\newcommand{\VoCmpI}[1][]{f_{#1}}
\newcommand{\VoCmpII}[1][]{g_{#1}}
\newcommand{\VoCmpIII}[1][]{h_{#1}}
\newcommand{\VoRen}[1][]{r_{#1}}
\newcommand{\VoRenI}[1][]{r_{#1}}
\newcommand{\VoRenII}[1][]{s_{#1}}

\newcommand{\VoSub}[1][]{s_{#1}}

\newcommand{\VoFunctor}[1][]{\mathcal{G}_{#1}}
\newcommand{\VoStruct}[1][]{\mathcal{S}_{#1}}
\newcommand{\VoProp}[1][]{\mathbb{D}_{#1}}
\newcommand{\VoOperad}[1][]{\mathbb{M}_{#1}}

\newcommand{\filler}[0]{\left(-\right)}

\newcommand{\subst}[1]{\left\{ #1 \right\}}
\newcommand{\assign}[2]{#2 / #1}

\newcommand{\subMVAR}[0]{s}

\newcommand{\SUBSYMB}[0]{\mathsf{Sub}}
\newcommand{\SUB}[2]{\SUBSYMB\left(#1 \Rightarrow #2\right)}

\begin{document}
\begin{frontmatter}
  \title{Multicategorical Semantics for Untyped Effects} 						
\thanks[ALL]{This work was partially supported by Grant No. 2020145 from the United States-Israel Binational Science Foundation (BSF) and Grant No. 2876/25 from the Israeli Science Foundation (ISF).}   
  \author{Liron Cohen\thanksref{a}}
  \author{Ariel Grunfeld\thanksref{a}}
   \address[a]{Faculty of Computer and Information Science\\ Ben-Gurion University\\	
    Be'er Sheva, Israel}  							
\begin{abstract}
Completeness proofs in categorical semantics usually proceed by building a syntactic category whose composition is given by substitution. For untyped effectful Call-by-Value languages, this runs into a basic obstacle: there is no canonical notion of simultaneous substitution of computations, since evaluation order is semantically meaningful.
We address this by taking single-computation substitutions, that is, binding steps, as primitive, and representing computation substitution as finite \emph{sequential} lists composed by concatenation. 
We formalize this idea in a one-object Freyd-multicategorical setting. We introduce \emph{Freyd operads}, separating a cartesian operad of values from a symmetric Ren-cartesian preoperad of computations, connected by a Freyd functor, and from any Freyd operad we construct a corresponding \emph{Freyd PROP of substitutions}.
We prove that this construction is representable and, in the strict one-object setting, left adjoint to restriction to codomain $1$. Using the induced term model, we interpret untyped computational $\lambda$-calculus with procedures and higher-order functions in weakly closed Freyd operads, and prove soundness, initiality, and completeness. This yields a categorical semantics tailored to untyped effectful computation and broad enough to encompass realizability-oriented models such as monadic combinatory algebras.
\end{abstract}
\end{frontmatter}

\section{Introduction}\label{intro}

Computational $\lambda$-calculus, introduced by Moggi~\cite{moggi1991notions}, provides an abstract language for effectful computation and has long served as a central test case for categorical semantics. For the typed language, the semantic story is well developed: from Moggi's $\lambda_{c}$-models~\cite{moggi1991notions},  through abstract Kleisli categories~\cite{fuhrmann1999direct} and Freyd categories~\cite{levy2003modelling}, completeness is typically established by constructing a \emph{term model}, namely a category whose objects are types and whose morphisms are terms modulo the equational theory. Completeness then follows because the valid equations in the term model are exactly the equations of the theory.

This type-based framework, however, does not translate naturally to an untyped setting. To interpret terms with several free variables, one needs product objects to interpret contexts, and this in turn requires the language itself to contain suitable product types. In typed settings this is natural, but for an untyped language there is no non degenerate type structure in which to add such products. Thus, the standard types as objects construction is poorly suited to untyped completeness arguments.

A more flexible alternative is to use a term model in which contexts are the objects and simultaneous substitutions between contexts are the morphisms~\cite{pitts2001categorical}. 
This subsumes the types and terms approach: a type is represented by a one variable context, a term of that type is represented by a substitution into that context, and equations between terms become equations between the corresponding substitutions. In pure languages, simultaneous substitutions carry the structure of finite tensor products, given by concatenation. By using contexts as objects, one therefore avoids the need for product types altogether, making this perspective particularly natural for untyped languages.

However, this approach breaks down for effectful languages, because there is no canonical notion of simultaneous substitution for effectful computations. The reason is that effectful computation is inherently sequential, and any would be notion of simultaneous computation substitution must choose an evaluation order.
For pure values, genuine \emph{simultaneous} substitution is well behaved: if
$
\subMVAR \;=\; \lPar{\,\assign{\VxVar[1]}{\VxVal[1]} \lSep \cdots \lSep \assign{\VxVar[n]}{\VxVal[n]}\,}
$
then applying $\subMVAR$ to a term replaces each free occurrence of $\VxVar[i]$ by $\VxVal[i]$ in one step, and such substitutions compose in the usual way.
For effectful computations, one might hope for an analogous notion of simultaneous substitution, replacing several \emph{computation holes} at once. The obstruction is that computations are only used by \emph{running} them, via the sequencing binder $\TxLetSMB$.
Concretely, consider the following computation contexts with two holes
{
\setlength{\abovedisplayskip}{3pt}
\setlength{\belowdisplayskip}{3pt}
\setlength{\abovedisplayshortskip}{0pt}
\setlength{\belowdisplayshortskip}{0pt}
\begin{align*}
    K_{1}[\filler,\filler] &\;:=\; \TxLet{\VxVar[1]}{\filler}{\TxLet{\VxVar[2]}{\filler}{\TxRet{\VxVar[1]}}}\\
    K_{2}[\filler,\filler] &\;:=\; \TxLet{\VxVar[1]}{\filler}{\TxLet{\VxVar[2]}{\filler}{\TxRet{\VxVar[2]}}}
\end{align*}
}
If we wish to substitute computations $\VxCmp[1]$ and $\VxCmp[2]$ into the two holes, we must choose an evaluation order, producing either
$
K_{1}[\VxCmp[1],\VxCmp[2]]
\;=\;
\TxLet{\VxVar[1]}{\VxCmp[1]}{\TxLet{\VxVar[2]}{\VxCmp[2]}{\TxRet{\VxVar[1]}}}
$
or
$
K_{2}[\VxCmp[2],\VxCmp[1]]
\;=\;
\TxLet{\VxVar[1]}{\VxCmp[2]}{\TxLet{\VxVar[2]}{\VxCmp[1]}{\TxRet{\VxVar[2]}}}.
$
In the presence of effects, these two programs need not be equivalent. Sequencing is part of the meaning, and there is therefore no canonical simultaneous computation substitution operation that forgets this order.

Thus, for untyped effectful languages, one encounters a genuine conundrum. On the one hand, the usual types as objects approach is unavailable, because there are no product types to interpret contexts. On the other hand, the contexts and substitutions approach also fails, because computations do not admit a canonical simultaneous substitution operation. The right replacement is suggested by the analysis above: single computation substitutions, that is, individual binding steps, should be taken as primitive, and general computation substitutions should be represented by finite lists of such single substitutions. Composition is then given by concatenation, namely by sequencing the binding steps.

A natural setting in which this phenomenon already appears is that of Freyd multicategories~\cite{staton2013universal}. There, values and computations are separated from the outset, and computations carry an order sensitive notion of substitution. This works just as well for untyped languages as for typed ones, and it leads to a straightforward completeness argument. 
At the same time, Freyd multicategories remain relatively unfamiliar, and using them directly as semantic targets obscures the connection with more standard categorical semantics~\cite{staton2014freyd}.
Previous suggestions for categorical semantics of untyped computational $\lambda$-calculus appear in~\cite{moggi1988computational,de2020untyped}, but without the completeness result obtained here.

The point of this paper is that one need not choose between the two viewpoints. We take the multicategorical semantics as the correct structural starting point, and from it derive a \emph{categorical} substitution construction tailored to the untyped effectful setting. This construction is semantic rather than syntax dependent, and can therefore be applied uniformly to arbitrary Freyd operadic data. It also extends a familiar pattern from ordinary multicategory theory, namely the free monoidal category generated by a multicategory. In the present sequential setting, the resulting category is no longer one of simultaneous substitutions, but one of ordered lists of single substitutions. In particular, in the strict one object setting relevant here, our construction yields the left adjoint to restriction to codomain $1$, providing the categorical counterpart of Freyd multicategorical substitution for untyped effectful programs.

Beyond the categorical issue itself, our motivation comes from effectful realizability, in particular from recent work on syntactic effectful realizability and evidence frame semantics~\cite{effhol,EF,yamada2026effectful}. 
For realizability with effects, one wants an untyped computational core language whose terms can serve as realizers, together with a completeness theorem strong enough to interpret those realizers in algebraic or combinatory models such as monadic combinatory algebras~\cite{cohen2025partial}.
The semantics developed here are designed with exactly this goal in mind: they provide a flexible untyped Call-by-Value core together with a class of categorical models broad enough to encompass effectful combinatory structures.

\paragraph*{Main Contributions:}
\begin{itemize}
    \item We introduce \emph{Freyd operads} as the one object structures that separate pure values from effectful computations in the untyped setting, together with \emph{weakly closed} structure for abstraction and application.
    \item From a Freyd operad, we construct a category of substitutions whose morphisms are finite lists of single substitutions and renamings, modulo the appropriate structural equations, yielding a Freyd PROP of substitutions.
    \item We prove that this construction is representable and left adjoint to restriction to codomain $1$, thereby recovering a categorical semantics from Freyd operadic data in the strict untyped setting.
    \item We interpret the untyped computational language $\xCLC$ in weakly closed Freyd operads and prove soundness, initiality of the term model, and completeness.
    \item We relate the resulting semantics to examples arising from monadic combinatory algebras, restriction categorical applicative systems, and reflexive object style models.
\end{itemize}

\textbf{Supplementary Material:}
The paper was autoformalized in Agda using Claude Fable 5.1.
The code is available at~\cite{agda-formalization}.

\textbf{Outline:}
The paper is organized as follows. 
Section~\ref{sec:background}  introduces $\xCLC$ and its equational theory and recalls the categorical background.
Section~\ref{sec:freyd} develops Freyd operads and the free substitution construction, culminating in representability and the adjunction with Freyd PROPs. 
Section~\ref{sec:semantics} gives the interpretation of $\xCLC$ in weakly closed Freyd operads and proves soundness, initiality, and completeness. 
Section~\ref{sec:examples} discusses examples and connections to  realizability motivated and categorical models.
Detailed proofs and constructions can be found in Appendix~\ref{sec:app}.

\section{Background: Untyped Computational Lambda Calculus}\label{sec:background}

This section recalls the category of renamings,   the syntax and equational theory of $\xCLC$, and the categorical vocabulary used in our semantics, in a concrete $\mathbb{N}$-indexed form that treats a number as a context size.

\subsection{Renamings}

Throughout, we use $n,m,k$ for natural numbers and write $[n]=\{1,\dots,n\}$ (with $[0]=\emptyset$).
\begin{defn}[Renaming]\label{def:op-ren}
A \emph{renaming} $r:m\to n$ is a function $[m]\to[n]$.
We write $\oREN{m}{n}$ for the set of such renamings. 
A bijective renaming is a \emph{permutation}, and we write $\oPERM{n}\subseteq \oREN{n}{n}$.
\end{defn}
We identify a renaming $r:[m]\to[n]$ with the tuple $(r(1),\ldots,r(m))$.
The identity is $\oId[n]=(1,\ldots,n)$.
Composition is ordinary composition of functions, and the tensor (coproduct) of
$r\in \oREN{m_1}{n_1}$ and $r'\in \oREN{m_2}{n_2}$ is the block sum
\vspace{-0.3em}
\[
r+r' \;:=\; \bigl(r(1),\ldots,r(m_1),\; n_1+r'(1),\ldots,n_1+r'(m_2)\bigr)\in \oREN{m_1+m_2}{n_1+n_2}.\vspace{-0.3em}
\]
If $r\in \oREN{m}{n}$ and $r'\in \oREN{k}{n}$ we write $[r,r']\in \oREN{m+k}{n}$ for their concatenation.

We single out the following basic renamings for swap, discard and copy:\vspace{-0.3em}
\begin{align*}
\oSym[m,n] :=(n+1,\ldots,n+m,\,1,\ldots,n)\in \oREN{m+n}{n+m}
\\
\oDiscard[n] := ()\in \oREN{0}{n}
\qquad,
\qquad
\oCopy[n] := (1,\ldots,n,\,1,\ldots,n)\in \oREN{2n}{n}.
\end{align*}

\begin{lemma}[\cite{burroni1991higher,lafont2003towards}]\label{lemma:ren-gen}
Every renaming $\VoRen \in \oREN{m}{n}$ is generated by composition and tensor product from
    $\oId[1]$, $\oSym[1,1]$, $\oCopy[1]$, and $\oDiscard[1]$, and every permutation $\VoRen \in \oPERM{n}$ is generated by composition and tensor product from
    $\oId[1]$ and $\oSym[1,1]$.
\end{lemma}

\begin{lemma}[\cite{burroni1991higher,lafont2003towards}]\label{lemma:ren-cat}
$\oRENSMB$ forms a cocartesian category with initial object $0$ and coproduct on objects given by $+$.
Moreover, $\oPERMSMB$ is a subgroupoid of $\oRENSMB$.
\end{lemma}

\subsection{Untyped Computational Lambda Calculus}
We work with an untyped, Call-by-Value variant of Moggi's computational $\lambda$-calculus, similar to $\xCLC$ in~\cite{sabry1997reflection}. 
The syntax separates \emph{values}, which do not perform effects, from \emph{computations}, which may perform effects. 
The term former $\TxRet{}$ injects values into computations, and $\TxLet{x}{M_1}{M_2}$ sequences computations by running $M_1$ and binding its returned value to $x$ in $M_2$.
As is standard, we identify terms up to $\alpha$-equivalence. All substitutions are capture-avoiding, and we freely rename bound variables to avoid capture.

The syntax of $\xCLC$ is given in the top part of Fig.~\ref{fig:clc-syntax}. 
The language is parameterized by a signature, which is a pair $\xSign = \left( \xFuncSMB , \xProcSMB \right)$, with two families of sets, parametrized by $\N$, where for each $n \in \N$, $\xFuncSymb{n}$ is a set of function symbols of arity $n$, and $\xProcSymb{n}$ is a set of procedure symbols of arity $n$.
We assume a countably infinite set $\xVAR$ of variables.
Values are either variables, the application of a function symbol of arity $n$ on $n$ values, or lambda abstractions, binding a variable in a computation. 
Computations are either a lifted value, a binding  of a computation  to a variable  within another computation
, the application of a procedure symbol of arity $n$ on $n$ values, or application of a value to another value.
Note that application could be treated as a procedure symbol, but we give it a dedicated rule, so we make it a separate constructor.
We track free variables using \emph{contexts}, which are finite lists of distinct variables
$\Gamma = x_{1},\ldots,x_{n}$.
We write $\Gamma,\VxVar,\Gamma'$ for context concatenation with $\VxVar$ in the middle.
From now on, terms are always considered relative to a context listing all of their free variables.

Because $\xCLC$ is untyped, there is no typing judgment.
Instead, there are two well-formedness judgments, whose rules are given in the bottom part of Fig.~\ref{fig:foclc-wf}: 
$\JxWFV{\Gamma}{\VxVal}$ for values and $\JxWFC{\Gamma}{\VxCmp}$ for computations, each requiring that all free variables occur in $\Gamma$.
We write $\JxWFT{\Gamma}{\VxTrm}$ for the appropriate judgment depending on whether $\VxTrm$ is a value or a computation.

To define the theory and semantics of $\xCLC$, we use substitutions that replace free variables with terms.
Since $\xCLC$ is Call-by-Value, variables range over values, so only values may be substituted for variables. Computations can enter only via sequencing with $\TxLet{-}{-}{-}$.
Although the equational theory uses only single-variable substitution, it is convenient to use simultaneous substitutions, which we therefore define for values.
This restriction is also conceptually important: there is no well-behaved notion of simultaneous substitution by computations, since computations interact through sequencing and evaluation order.

\begin{defn}[Simultaneous Substitution]
Let $\VxCxt = \VxVar[1] , \ldots , \VxVar[n]$ be a context.
A \emph{(simultaneous) substitution into $\VxCxt$} is a list
$\VoSub =\assign{\VxVar[1]}{\VxVal[1]}  , \ldots , \assign{\VxVar[n]}{\VxVal[n]}$
such that each $\VxVal[i]$ is a well-formed value in some context $\VxCxt[i]$ (i.e.\ $\JxWFV{\VxCxt[i]}{\VxVal[i]}$).
We view $\VoSub$ as a substitution \emph{from} the list of source contexts $\Gamma_1,\ldots,\Gamma_n$ \emph{to} $\Gamma$, and write $
\VoSub \in \SUB{\Gamma_1,\ldots,\Gamma_n}{\Gamma}$.
Given $\VoSub=\assign{x_{1}}{V_{1}},\ldots,\assign{x_{n}}{V_{n}} \in \SUB{\Gamma_1,\ldots,\Gamma_n}{\Gamma}$
and a well formed term $\JxWFT{\Gamma}{t}$, we write $t\subst{\VoSub}$ for the result of applying $\VoSub$ to $t$.
It is defined by structural recursion on $t$, replacing each free occurrence of $x_i$ by $V_i$ (and avoiding capture in the $\lambda$-case).
\end{defn}

\begin{lemma}\label{lemma:subst}
If $\VoSub \in \SUB{\VxCxt[1]}{\VxCxt[2]}$ and $\JxWFT{\VxCxt[2]}{\VxTrm}$, then
$\JxWFT{\VxCxt[1]}{\VxTrm \subst{\VoSub}}$.
\end{lemma}

\begin{figure}[t]
\centering

{
\setlength{\tabcolsep}{3pt}
\renewcommand{\arraystretch}{1.05}
\begin{tabular}{l@{\qquad}r@{\;\;}c@{\;\;}l}
\textbf{Terms} & $\VxTrm$ & ::= &
$\VxVal \;\mid\; \VxCmp$
\\
\textbf{Values} & $\VxVal$ & ::= &
$\VxVar \;\mid\; \VxFunc(\VxVal,\ldots,\VxVal) \;\mid\; \TxAbs{\VxVar}{\VxCmp}$
\\
\textbf{Computations} & $\VxCmp$ & ::= &
$\TxRet{\VxVal} \;\mid\; \TxLet{\VxVar}{\VxCmp}{\VxCmp} \;\mid\; \VxProc(\VxVal,\ldots,\VxVal) \;\mid\; \TxApp{\VxVal}{\VxVal}$
\end{tabular}
}

\vspace{0.8ex}
\hrule height 0.1pt
\vspace{1.5ex}

{
\setlength{\tabcolsep}{10pt}
\renewcommand{\arraystretch}{1.25}
\begin{tabular}{ccc}
  \infer{ \JxWFV{\VxVar}{\VxVar} }{}
  &
  \infer{ \JxWFV{\VxCxt[1] , \ldots , \VxCxt[n]}{ \VxFunc\!\left(\VxVal[1] , \ldots , \VxVal[n]\right) } }
        { \VxFunc \in \xFuncSymb{n} \;\; &\;\; \forall i \in \{1,\ldots,n\}.\, \JxWFV{\VxCxt[i]}{ \VxVal[i] } }
  &
  \infer{ \JxWFV{\VxCxt}{ \TxAbs{\VxVar}{\VxCmp} } }{ \JxWFC{ \VxCxt , \VxVar }{ \VxCmp } }
  \\[1.6ex]
  \infer{ \JxWFC{\VxCxt}{ \TxRet{\VxVal} } }{ \JxWFV{ \VxCxt }{ \VxVal } }
  &
  \infer{ \JxWFC{\VxCxt[1] , \ldots , \VxCxt[n]}{ \VxProc\!\left(\VxVal[1] , \ldots , \VxVal[n]\right) } }
        { \VxProc \in \xProcSymb{n} \;\; &\;\; \forall i \in \{1,\ldots,n\}.\, \JxWFV{\VxCxt[i]}{ \VxVal[i] } }
  &
  \infer{ \JxWFC{\VxCxt[1] , \VxCxt[2]}{ \TxLet{\VxVar}{\VxCmp[2]}{\VxCmp[1]} } }
        { \JxWFC{ \VxCxt[1] , \VxVar }{ \VxCmp[1] } \;\; &\;\; \JxWFC{\VxCxt[2]}{ \VxCmp[2] } }
\end{tabular}

\vspace{1.2ex}

\setlength{\tabcolsep}{8pt}
\renewcommand{\arraystretch}{1.15}
\begin{tabular}{cccc}
  \infer{ \JxWFC{\VxCxt[1] , \VxCxt[2]}{\TxApp{\VxVal[1]}{\VxVal[2]}} }
        { \JxWFV{\VxCxt[1]}{\VxVal[1]} \;\; &\;\; \JxWFV{\VxCxt[2]}{\VxVal[2]} }
&
  \infer{\JxWFT{\VxCxt , \VxVar}{\VxTrm}}{ \JxWFT{\VxCxt}{\VxTrm} & \VxVar \notin \VxCxt }
&
  \infer{ \JxWFT{\VxVar[\VoRen\left(1\right)] , \ldots , \VxVar[\VoRen\left(n\right)]}{\VxTrm} }
        { \JxWFT{\VxVar[1] , \ldots , \VxVar[n]}{\VxTrm} & \VoRen \in \oPERM{n} }
&
  \infer{\JxWFT{\VxCxt , \VxVarII}{\VxTrm \subst{\assign{\VxVarI[1]}{\VxVarII} , \assign{\VxVarI[2]}{\VxVarII} }}}
        { \JxWFT{\VxCxt , \VxVarI[1] , \VxVarI[2]}{\VxTrm} }
\end{tabular}
}

\caption{Syntax and well-formedness rules for $\xCLC$.}
\label{fig:lambdaML-syntax-wf}
\label{fig:clc-syntax}
  \label{fig:foclc-wf}
      \label{tab:foclc-wf}
\end{figure}

The theory of $\xCLC$ is presented in Fig.~\ref{fig:clc-theory}. 
The equational theory is the least congruence closed under the term formers, generated by the monad laws for $\TxRet{-}$ and $\TxLet{-}{-}{-}$. Intuitively, (lunit) and (runit) say that $\TxRet{-}$ is the unit for sequencing, and (assoc) states associativity of sequencing.

\subsection{Freyd Categories}

We use Freyd-style categorical semantics to model Call-by-Value effectful computation~\cite{levy2003modelling}. 
Intuitively, the pure fragment (values) forms a cartesian theory of renamings, copying, and discarding, while computations form a premonoidal category where sequencing is explicit. We package both components in a Freyd structure via a functor $J : \oVal \to \oCmp$ that embeds pure maps as central morphisms.
Throughout, a natural number $n$ is read as a context of $n$ variables. A morphism $m\to n$ should be read as an $n$-tuple of effectful bindings that may use $m$ variables as input. With this convention, addition $m+n$ represents context concatenation, and whiskering by $k$ on the left or right extends a morphism with $k$ unused variables.
The definitions below are standard~\cite{power1997premonoidal,pirashvili2002prop,bonchi2016lawvere}. We spell them out in the strict $\N$-indexed form used later in the paper.

\begin{defn}[Pre-PROP]
A pre-PROP $\VoProp$ is a strict symmetric premonoidal category with the natural numbers as objects and addition as tensor product.
We write $k\oSL(-)$ and $(-)\oSR k$ for the left and right whiskering endofunctors induced by tensoring with $k$,
and $\oSym[m,n]:m+n\to n+m$ for the symmetry isomorphisms.

\end{defn}

\begin{defn}[Centrality in a Pre-PROP]
Given a pre-PROP $\VoProp$, a morphism $\VoCmpI \in \oHom{\VoProp}{m}{n}$ is \emph{central}, if for every $\VoCmpII \in \oHom{\VoProp}{m'}{n'}$, $\VoCmpI \oSR n' \circ m \oSL \VoCmpII = n \oSL \VoCmpII \circ \VoCmpI \oSR m'$ and $\VoCmpII \oSR n \circ m' \oSL \VoCmpI = n' \oSL \VoCmpI \circ \VoCmpII \oSR m$.
\end{defn}

\begin{figure}[t]
\centering
\makebox[\linewidth][c]{\hspace*{-0.6cm}
\begin{minipage}{\linewidth}\centering
\setlength{\tabcolsep}{2pt}
\renewcommand{\arraystretch}{1.35}
    \begin{tabular}{ccc}
      $\infer[\text{\tiny (id)}]{\JxEq{\VxCxt}{\VxTrm}{\VxTrm}}{}$
      &
      $\infer[\text{\tiny (sym)}]{\JxEq{\VxCxt}{\VxTrm[2]}{\VxTrm[1]}}{\JxEq{\VxCxt}{\VxTrm[1]}{\VxTrm[2]}}$
      &
      $\infer[\text{\tiny (trans)}]{\JxEq{\VxCxt}{\VxTrm[1]}{\VxTrm[3]}}{\JxEq{\VxCxt}{\VxTrm[1]}{\VxTrm[2]} & \JxEq{\VxCxt}{\VxTrm[2]}{\VxTrm[3]}}$
    \end{tabular}
    \\[2ex]

    \begin{tabular}{cc}
      $\infer[\text{\tiny (ret)}]{\JxEqCmp{\VxCxt}{\TxRet{\VxVal}}{\TxRet{\VxVal'}}}{\JxEqVal{\VxCxt}{\VxVal}{\VxVal'}}$
      &
      $\infer[\text{\tiny (f-app)}]{ \JxEqVal{\VxCxt[1] , \ldots , \VxCxt[k]}{\VxFunc\left( \VxVal[1] , \ldots , \VxVal[k]\right)}{\VxFunc\left( \VxVal[1]' , \ldots , \VxVal[k]'\right)} }{\forall i \in \{1, \ldots , k\} . \; \JxEqVal{\VxCxt[i]}{\VxVal[i]}{\VxVal[i]'} }$
    \end{tabular}
    \\[2ex]

    \begin{tabular}{cc}
      $\infer[\text{\tiny (let)}]{
        \JxEqCmp{\VxCxt[1] , \VxCxt[2]}
          {\TxLet{\VxVar}{\VxCmp[2]}{\VxCmp[1]}}
          {\TxLet{\VxVar}{\VxCmp[2]'}{\VxCmp[1]'}}
      }{
        \JxEqCmp{\VxCxt[1] , \VxVar}{\VxCmp[1]}{\VxCmp[1]'}
        &
        \JxEqCmp{\VxCxt[2]}{\VxCmp[2]}{\VxCmp[2]'}
      }$
      &
      $\infer[\text{\tiny (p-app)}]{ \JxEqCmp{\VxCxt[1] , \ldots , \VxCxt[k]}{\VxProc\left( \VxVal[1] , \ldots , \VxVal[k]\right)}{\VxProc\left( \VxVal[1]' , \ldots , \VxVal[k]'\right)} }{\forall i \in \{1, \ldots , k\} . \; \JxEqVal{\VxCxt[i]}{\VxVal[i]}{\VxVal[i]'} }$
    \end{tabular}
    \\[2ex]

    \begin{tabular}{cc}
      $\infer[\text{\tiny (lunit)}]{\JxEqCmp{\VxCxt[1], \VxCxt[2]}{\TxLet{\VxVar}{\TxRet{\VxVal}}{\VxCmp}}{\VxCmp \subst{\assign{\VxVar}{\VxVal}} } }{ \JxWFC{\VxCxt[1] , \VxVar}{ \VxCmp } & \JxWFV{\VxCxt[2]}{\VxVal} }$
      &
      $\infer[\text{\tiny (runit)}]{\JxEqCmp{\VxCxt}{\TxLet{\VxVar}{\VxCmp}{\TxRet{\VxVar}}}{\VxCmp}}{ \JxWFC{\VxCxt}{\VxCmp} }$
    \end{tabular}
    \\[2ex]

    \begin{tabular}{c}
      $\infer[\text{\tiny (assoc)}]{\JxEqCmp{\VxCxt[1] , \VxCxt[2] , \VxCxt[3]}{\TxLet{\VxVar[2]}{\left(\TxLet{\VxVar[1]}{\VxCmp[1]}{\VxCmp[2]}\right)}{\VxCmp}}{\TxLet{\VxVar[1]}{\VxCmp[1]}{\TxLet{\VxVar[2]}{\VxCmp[2]}{\VxCmp}}}}{ \JxWFC{\VxCxt[1] , \VxVar[2]}{\VxCmp}
      & \JxWFC{\VxCxt[2] , \VxVar[1]}{\VxCmp[2]}
      & \JxWFC{\VxCxt[3]}{\VxCmp[1]}
      }$
    \end{tabular}
    \\[3ex]

    \begin{tabular}{ccc}
          {$\infer[\text{\tiny (beta)}]{ \JxEqCmp{\VxCxt[1] , \VxCxt[2]}{\TxApp{\left( \TxAbs{\VxVar}{\VxCmp}\right)}{\VxVal}}{\VxCmp \subst{\assign{\VxVar}{\VxVal}}} }{ \JxWFC{\VxCxt[1] , \VxVar}{ \VxCmp } & \JxWFV{ \VxCxt[2] }{ \VxVal } }$}
          &
      {$\infer[\text{\tiny (abs)}]{\JxEqVal{\VxCxt}{\TxAbs{\VxVar}{\VxCmp}}{\TxAbs{\VxVar}{\VxCmp'}}}{ \JxEqCmp{\VxCxt , \VxVar}{ \VxCmp }{\VxCmp'} }$}
      &
      {$\infer[\text{\tiny (app)}]{ \JxEqCmp{\VxCxt[1] , \VxCxt[2]}{\TxApp{\VxVal[1]}{\VxVal[2]}}{\TxApp{\VxVal[1]'}{\VxVal[2]'}}}{ \JxEqVal{\VxCxt[1]}{\VxVal[1]}{\VxVal[1]'} & \JxEqVal{\VxCxt[2]}{\VxVal[2]}{\VxVal[2]'} }$}
    \end{tabular}
\end{minipage}}
  \caption{Equational theory of $\xCLC$.}
  \label{fig:clc-theory}
\end{figure}

\begin{defn}[PROP]
A PROP is a pre-PROP where all morphisms are central.
\end{defn}

\begin{defn}[Cartesian PROP]
\label{def:cart-PROP}
A PROP $\VoProp$ is \emph{cartesian} if it is equipped with natural transformations: copy, $\oCopy[n] : n \to 2n$ and discard, $\oDiscard[n] : n \to 0$, 
forming a commutative comonoid for each $n$, with the following coherence condition:
$
\oCopy[n_{1}+n_{2}]
\;\equiv\;
\bigl(n_{1} + \oSym[n_{1},n_{2}] + n_{2}\bigr) \circ
\bigl(\oCopy[n_{1}] + \oCopy[n_{2}]\bigr) 
$ and $
\oDiscard[n_{1}+n_{2}]
\;\equiv\;
\oDiscard[n_{1}] + \oDiscard[n_{2}]
$.
\end{defn}

Our arrow convention reads $m \to n$ as producing $n$ outputs from $m$ inputs (substitution from an input context to an output context).
This is the standard Lawvere theory presentation~\cite{lawvere1963functorial}.

\begin{defn}[Freyd PROP]
A Freyd PROP is a triple $( \oVal , \oCmp , \oReturn )$ where $\oVal$ is a cartesian PROP, $\oCmp$ is a pre-PROP, and
$\oReturn : \oVal \longrightarrow \oCmp$ is a pre-PROP functor such that 
for every $m,n\in\N$ and every
$\VoVal \in \oHom{\oVal}{m}{n}$, the morphism $\oReturn(\VoVal)$ is central in $\oCmp$.
\end{defn}

\begin{defn}[Structure-Preserving Functors]\label{def:prop-functors}
~
\begin{itemize}
\item A \emph{pre-PROP functor} $F:\VoProp[1]\to\VoProp[2]$ is a strict symmetric premonoidal functor that is the identity on objects
(hence preserves addition on objects), and preserves whiskering and symmetries.%
\item A \emph{cartesian PROP functor} is a pre-PROP functor that also preserves copy and discard:
$F(\oCopy[n])=\oCopy[n]$ and $F(\oDiscard[n])=\oDiscard[n]$ for all $n$.
\item A \emph{Freyd PROP functor} between Freyd PROPs $(\oVal[1],\oCmp[1],\oReturn[1])\to(\oVal[2],\oCmp[2],\oReturn[2])$
is a pair $(F^{\oVal},F^{\oCmp})$ where $F^{\oVal}$ is a cartesian PROP functor, $F^{\oCmp}$ is a pre-PROP functor,
and 
$
F^{\oCmp}\bigl(\oReturn[1](\VoVal)\bigr) \;=\; \oReturn[2]\bigl(F^{\oVal}(\VoVal)\bigr)$ 
 for all $\VoVal$.
 \end{itemize}
\end{defn}

\section{Freyd operads and free Freyd PROPs}\label{sec:freyd}
This section introduces the Freyd-operadic structures used in our semantics and the associated substitution construction. 
A Freyd operad consists of a cartesian operad of values, a symmetric Ren-cartesian preoperad of computations, and a cartesian functor embedding values into computations. 
From any Freyd operad we construct a corresponding Freyd PROP of substitutions, whose morphisms are finite lists of elementary substitution and renaming steps, quotiented by the equations forced by the operadic structure. 

Sections~\ref{subsec:preop} and~\ref{subsec:freyd-operads} provide the required one-object, $\mathbb{N}$-indexed background.
The main novel material starts in Section~\ref{subsec:subs}, where we define the substitution construction, and continues in Section~\ref{subsec:presheaves} (representability) and Section~\ref{subsec:adjunction} (the adjunction).

\subsection{Preoperads and reindexing}\label{subsec:preop}

Operads and multicategories provide a standard account of algebraic structure in terms of operations with many inputs and one output, with composition corresponding to \emph{parallel} substitution.
For effectful Call-by-Value computation, however, the basic structural operation is \emph{sequencing}: substituting a computation into a hole is inherently ordered, and evaluation order is semantically meaningful.
Following Staton--Levy's premulticategorical treatment of impure languages~\cite{staton2013universal}, and specializing to the untyped setting, we work with the strict $\mathbb{N}$-indexed one-object form of premulticategories. Although the notions in this subsection are standard there, we spell them out to fix notation for the free substitution construction of Sec.~\ref{subsec:subs}.
Since we work in an untyped setting, a morphism in arity $n$ is simply an element of a family $\oC{n}$, read as an $n$-ary operation.
Composition is given by substituting an operation into a designated input position.
To model renaming of free variables, we equip preoperads with a coherent reindexing action by renamings.

\begin{defn}[Preoperad]\label{def:preoperad}
A preoperad is given by the following data:
\begin{itemize}
    \item For each $n \in \N$, a collection $\oC{n}$ of \emph{morphisms}

    \item \emph{Identity morphism}: A morphism $\oId \in \oC{1}$

    \item \emph{Composition}: Given a morphism $\VoCmpII \in \oC{m}$ and a morphism $\VoCmpI \in \oC{n_{1}+1+n_{2}}$, there is a morphism $\oSubAp{\VoCmpI}{ n_{1} \oSL g \oSR n_{2}} \in \oC{n_{1} + m + n_{2}}$
\end{itemize}
satisfying the following laws:
\begin{enumerate}
    \item \emph{Left Unit}. For all $\VoCmp \in \oC{m}$: 
    $ \oSubAp{\oId}{ \VoCmp } \equiv \VoCmp $

    \item \emph{Right Unit}. For all $\VoCmp \in \oC{n_{1} + 1 + n_{2}}$: 
    $  \oSubAp{\VoCmp}{n_{1} \oSL \oId \oSR n_{2}} \equiv \VcCmp $

    \item \emph{Associativity}. For all $\VoCmpI \in \oC{m_{1} + 1 + m_{2}}$, $\VcCmpII \in \oC{n_{1} + 1 + n_{2}}$, and $\VcCmpIII \in \oC{n}$:
    
    $ \oSubAp{\oSubAp{\VoCmpI}{m_{1} \oSL \VoCmpII \oSR m_{2}}}{m_{1} + n_{1} \oSL \VoCmpIII \oSR n_{2} + m_{2}}
    \equiv
    \oSubAp{\VoCmpI}{ m_{1} \oSL \oSubAp{\VoCmpII}{n_{1} \oSL \VoCmpIII \oSR n_{2}} \oSR m_{2} }
    $
\end{enumerate}
\end{defn}
Concretely, if $\VoCmpI \in \oC{m_{1}+1+m_{2}}$ and $\VoCmpII \in \oC{m}$, then
$\oSubAp{\VoCmpI}{m_{1}\oSL \VoCmpII \oSR m_{2}} \in \oC{m_{1}+m+m_{2}}$ is obtained by inserting $\VoCmpII$
into the distinguished input of $\VoCmpI$ (the $(m_{1}+1)$st input), leaving the other inputs untouched.
In terms of multicategories, one may think of this as composing $\VoCmpI$ with a tuple of arguments in which all
positions are filled by identities except the distinguished one, which is filled by $\VoCmpII$:\vspace{-0.3em}
\[\oSubAp{\VoCmpI}{m_{1}\oSL \VoCmpII \oSR m_{2}}
\;\approx\;
\VoCmpI\big(\underbrace{\oId,\ldots,\oId}_{m_{1}},\,\VoCmpII,\,\underbrace{\oId,\ldots,\oId}_{m_{2}}\big)\]
In the strict $\mathbb{N}$-indexed presentation this plays the role of \emph{whiskering with identities} to make arities line up for composition.
For example, take $\VoCmpI \in \oC{2+1+0}=\oC{3}$ and $\VoCmpII \in \oC{4}$.
Then $\oSubAp{\VoCmpI}{2 \oSL \VoCmpII \oSR 0} \in \oC{2+4+0}=\oC{6}$.
If $\VoCmpIII \in \oC{1}$, associativity specializes to the arity check
$\Big(\VoCmpI\{2 \dashv \VoCmpII \vdash 0\}\Big)\{2+1 \dashv \VoCmpIII \vdash 2\}
\;\equiv\;
\VoCmpI\{2 \dashv (\VoCmpII\{1 \dashv \VoCmpIII \vdash 2\}) \vdash 0\}$,
where both sides have arity $6$.

The terminology for cartesian premulticategories  in~\cite{staton2013universal} is not compositional: a cartesian multicategory is not simply a cartesian premulticategory that is also a multicategory, but requires additional naturality with respect to renamings. 
To keep the one-object presentation compositional, we explicitly separate the presence of a coherent renaming action ($\mathcal{S}$-cartesian) from the additional naturality conditions that recover fully cartesian (multi)categorical structure.

\begin{defn}[$\mathcal{S}$-Cartesian Preoperad]\label{def:scart-preop}
Given a subcategory $\mathcal{S}$ of $\oRENSMB$, closed under $+$ as a strict tensor product,
a preoperad $\oVal$ is $\mathcal{S}$-cartesian when for every  $\VoRen \in \oHom{\mathcal{S}}{m}{n}$ and  $\VoVal \in \oV{m}$, there is a morphism $\oRenAp{\VoVal}{ \VoRen } \in \oV{n}$, such that:
\begin{enumerate}
    \item $ \oRenAp{\VoVal}{\oId[n]} \equiv \VoVal  $
    \item $ \oRenAp{ \VoVal }{ \VoRen[2] \circ \VoRen[1]} \equiv \oRenAp{\oRenAp{\VoVal}{\VoRen[1]}}{\VoRen[2]} $
    \item  For all $\VoValI \in \oV{m_{1} + 1 + m_{2}}$, $\VoValII \in \oV{m}$, $\VoRenI[i] \in \oHom{\mathcal{S}}{m_{i}}{n_{i}}$, $\VoRenII \in \oHom{\mathcal{S}}{m}{n}$:
    
    $ \oRenAp{\oSubAp{\VoValI}{ m_{1} \oSL \VoValII  \oSR m_{2} }}{ \VoRenI[1] + \VoRenII + \VoRenI[2]  } \equiv \oSubAp{ \oRenAp{\VoValI}{\VoRenI[1] + \oId[1] + \VoRenI[2]}  }{ n_{1} \oSL \oRenAp{\VoValII}{\VoRenII} \oSR n_{2} } $
\end{enumerate}
\end{defn}

Intuitively, symmetric preoperads 
express that reindexing by permutations (exchange) interacts naturally with single substitution.

\begin{defn}[Symmetric Preoperad]\label{def:sym-preop}
A preoperad $\oCmp$ is called a \emph{symmetric preoperad} when it is $\oPERMSMB$-cartesian, and $\oSym[1,n]$ and $\oSym[n,1]$ are natural in $n$, i.e., for every $\VoCmpI \in \oC{m_{1} + 2 + m_{2}}$ and $\VoCmpII \in \oC{n}$:\vspace{-0.3em}
\begin{align*}
     \oSubAp{\oRenAp{\VoCmpI}{\oId[m_{1}] + \oSym[1,1] + \oId[m_{2}]}}{m_{1} \oSL \VoCmpII \oSR 1 + m_{2}}
    \equiv
    \oRenAp{\oSubAp{\VoCmpI}{m_{1} + 1 \oSL \VoCmpII \oSR m_{2}}}{ \oId[m_{1}] + \oSym[1,n] + \oId[m_{2}] }
\\
\oSubAp{\oRenAp{\VoCmpI}{\oId[m_{1}] + \oSym[1,1] + \oId[m_{2}]}}{m_{1} + 1 \oSL \VoCmpII \oSR m_{2}}
    \equiv
    \oRenAp{\oSubAp{\VoCmpI}{m_{1} \oSL \VoCmpII \oSR 1 + m_{2}}}{ \oId[m_{1}] + \oSym[n,1] + \oId[m_{2}] }
\end{align*}
\end{defn}

In the pure fragment, weakening and contraction are harmless: weakening adds an unused variable, and contraction merely identifies
two variables, both of which commute with substitution.
In contrast, in a Call-by-Value effectful language, weakening or contraction at the level of computations would implicitly allow
discarding or duplicating effects. 
For example, even if $x$ does not occur free in a computation $M_{1}$, the term
$\TxLet{x}{M_{2}}{M_{1}}$ is not equivalent to $M_{1}$ in general, since the effect of $M_{2}$ is still performed.
Exchange is different: it only reorders variables in the surrounding context and does not change what is evaluated or when.
Accordingly, our effectful structure requires naturality for permutations (exchange), while copying and discarding are confined
to the pure cartesian part.

Preoperad functors are the structure-preserving morphisms between preoperads, namely the one-object specializations of premulticategory morphisms from~\cite{staton2013universal}. In the strict $\mathbb{N}$-indexed presentation there is no separate object component, since arities already encode lists of the unique object. In the cartesian case, such functors must moreover preserve reindexing as well as substitution.

\begin{defn}[Preoperad Functor]

Given preoperads $\oCmp[1]$ and $\oCmp[2]$, a preoperad functor $\VoFunctor$ from $\oCmp[1]$ to $\oCmp[2]$, $\VoFunctor : \oCmp[1] \longrightarrow \oCmp[2]$, is an assignment of a morphism $\VoFunctor \VoCmp \in \oC[2]{n}$ for each $\VoCmp \in \oC[1]{n}$, such that 
    $\VoFunctor \oId \equiv \oId$ and 
    $\VoFunctor \left( \oSubAp{\VoCmpI}{ m_{1} \oSL \VoCmpII \oSR m_{2} } \right) \equiv \oSubAp{\left(\VoFunctor \VoCmpI\right)}{ m_{1} \oSL \VoFunctor \VoCmpII \oSR m_{2}} $. 
Together, these two properties are called \emph{functoriality}.
\end{defn}

\begin{defn}[Cartesian Preoperad Functor]
A preoperad functor $\VoFunctor$ between $\oRENSMB$-cartesian preoperads is called \emph{cartesian} if it preserves all renamings: $\VoFunctor \left(\oRenAp{\VoCmpI}{\VoRen}\right) \equiv \oRenAp{\left(\VoFunctor \VoCmpI\right)}{\VoRen}$
\end{defn}

\subsection{Centrality and Freyd operads}\label{subsec:freyd-operads}

Composition in a premulticategory is inherently sequential, since one substitutes into a chosen input position. Centrality identifies those morphisms for which the order of substitution does not matter, up to the evident arity adjustment. 
In general, multicategories are recovered as the central part of premulticategories~\cite{staton2013universal}. 
In our one-object, $\mathbb N$-indexed setting, this recovers operads from preoperads and isolates the fragment relevant to our semantics, where values are cartesian and computations admit only exchange.

\begin{defn}\label{def:op-centrality}
We say $\VoCmpII[1] \in \oC{n_{1}}$ \emph{commutes} with $\VoCmpII[2] \in \oC{n_{2}}$ when the following holds for every $\VoCmpI \in \oC{m_{1} + 1 + m + 1 + m_{2}}$:\vspace{-0.5em}
\begin{align*}
\oSubAp{\oSubAp{\VoCmpI}{m_{1} \oSL \VoCmpII[1] \oSR m + 1 + m_{2}}}{m_{1} + n_{1} + m \oSL \VoCmpII[2] \oSR m_{2}}
     \equiv  \oSubAp{\oSubAp{\VoCmpI}{m_{1} + 1 + m \oSL \VoCmpII[2] \oSR m_{2}}}{m_{1} \oSL \VoCmpII[1] \oSR m + n_{2} + m_{2}} 
\\[0.5em]
    \oSubAp{\oSubAp{\VoCmpI}{m_{1} \oSL \VoCmpII[2] \oSR m + 1 + m_{2}}}{m_{1} + n_{2} + m \oSL \VoCmpII[1] \oSR m_{2}}
     \equiv  \oSubAp{\oSubAp{\VoCmpI}{m_{1} + 1 + m \oSL \VoCmpII[1] \oSR m_{2}}}{m_{1} \oSL \VoCmpII[2] \oSR m + n_{1} + m_{2}}
\end{align*}
A morphism $\VoCmp$ is \emph{central} if it commutes with every morphism in $\oCmp$.
A preoperad $\oCmp$ is an \emph{operad} if every morphism in $\oCmp$ is central.
\end{defn}
If $f \in \oC{k}$ and $g_{1},\ldots,g_{k}$ are central, iterated single substitution is independent of insertion order, so we write $f\{g_{1}+\cdots+g_{k}\}$ for the resulting parallel composite.

Cartesian operads are 
operads equipped with renamings and the usual structural maps for exchange, discarding, and copying, satisfying the standard naturality axioms.

\begin{defn}[Cartesian Operad]\label{def:cart-op}
An operad $\oVal$ is said to be \emph{cartesian} when it is $\oRENSMB$-cartesian, symmetric, and discarding ($\oDiscard$) and copying ($\oCopy$) are natural in the following sense:
\begin{enumerate}
    \item For  $\VoCmpI \in \oV{m_{1} + m_{2}}$ and $\VoCmpII \in \oV{n}$:
    $ \oSubAp{\oRenAp{\VoCmpI}{\oId[m_{1}] + \oDiscard[1] + \oId[m_{2}]}}{ m_{1} \oSL \VoCmpII \oSR m_{2} } \equiv \oRenAp{\VoCmpI}{\oId[m_{1}] + \oDiscard[n] + \oId[m_{2}]}$
    
    \item For  $\VoCmpI \in \oV{m_{1} + 2 + m_{2}}$ and $\VoCmpII \in \oV{n}$:
    
    $ \oSubAp{\oRenAp{\VoCmpI}{\oId[m_{1}] + \oCopy[1] + \oId[m_{2}]}}{ m_{1} \oSL \VoCmpII \oSR m_{2} } \equiv \oRenAp{\oSubAp{\VoCmpI}{m_{1} \oSL \VoCmpII + \VoCmpII \oSR m_{2}}}{\oId[m_{1}] + \oCopy[n] + \oId[m_{2}]}$
\end{enumerate}
\end{defn}

For the semantics of $\xCLC$, we need two preoperads, one to interpret values and simultaneous substitutions, and the other to interpret computations and their binding structure given by the $\TxLet{-}{-}{-}$ form.
The preoperad designated for values is a cartesian operad, to enable the parallel, non-linear structure of simultaneous substitutions.
The preoperad designated for computations is symmetric $\oRENSMB$-cartesian, to enable the appropriate interaction of reindexing with the $\TxLet{-}{-}{-}$ form.
The $\TxRet{}$ form relates the two, and essentially embeds the structure of values within the syntax of computations. This embedding is enabled through a preoperad functor between the two preoperads, preserving centrality and the cartesian structure.
Crucially, the image of this functor is required to be central, expressing that pure computations commute with sequencing.
These three components together are called a \emph{Freyd operad}. 

\begin{defn}[Freyd Operad]\label{def:freyd-operad}
A Freyd operad is a triple $\left( \oVal , \oCmp , \oReturn \right)$ where $\oVal$ is a cartesian operad, $\oCmp$ is a symmetric $\oRENSMB$-cartesian preoperad, and $\oReturn : \oVal \longrightarrow \oCmp$ is a cartesian preoperad functor whose image is \emph{cartesian in $\oCmp$}, meaning that for every $\VoVal \in \oV{n}$:
\begin{enumerate}
\item \emph{Centrality}. $\oReturn \VoVal$ is central in $\oCmp$

\item \emph{Discardability}. For every $\VoCmpI \in \oC{m_{1}+m_{2}}$:

$ \oSubAp{\oRenAp{\VoCmpI}{\oId[m_{1}] + \oDiscard[1] + \oId[m_{2}]}}{ m_{1} \oSL \oReturn \VoVal \oSR m_{2} } \equiv \oRenAp{\VoCmpI}{\oId[m_{1}] + \oDiscard[n] + \oId[m_{2}]}$

\item \emph{Copyability}. For every $\VoCmpI \in \oC{m_{1}+2+m_{2}}$:

$ \oSubAp{\oRenAp{\VoCmpI}{\oId[m_{1}] + \oCopy[1] + \oId[m_{2}]}}{ m_{1} \oSL \oReturn \VoVal \oSR m_{2} } \equiv \oRenAp{\oSubAp{\VoCmpI}{m_{1} \oSL \oReturn \VoVal + \oReturn \VoVal \oSR m_{2}}}{\oId[m_{1}] + \oCopy[n] + \oId[m_{2}]}$
\end{enumerate}
\end{defn}

In the term model of $\xCLC$, the Freyd functor $\oReturn : \oVal \to \oCmp$ is interpreted by the constructor $\TxRet{-}$. Centrality of its image expresses that pure computations commute with sequencing. Thus inserting $\TxRet{V}$ into a surrounding let-context does not constrain evaluation order.
Concretely, for any computation $M$ and computation context $K[-]$,
the two programs obtained by evaluating $\TxRet{V}$ before or after $M$ are identified by the equational theory:\vspace{-0.3em}
\[
K\big[\TxLet{x}{M}{\TxLet{y}{\TxRet{V}}{M'}}\big]\;\equiv\;K\big[\TxLet{y}{\TxRet{V}}{\TxLet{x}{M}{M'}}\big]
\]
Discardability and copyability express
that a pure computation may be erased or duplicated without changing the effects performed.

\begin{defn}[Freyd Operad Functor]\label{def:freyd-map}
Given Freyd operads $\left( \oVal[1] , \oCmp[1] , \oReturn[1] \right)$ and $\left( \oVal[2] , \oCmp[2] , \oReturn[2] \right)$, a Freyd operad functor $\VoFunctor : \oReturn[1] \Longrightarrow \oReturn[2]$ is a pair $\VoFunctor = \left( \VoFunctor^{\oVal} , \VoFunctor^{\oCmp} \right)$ 
such that $\VoFunctor^{\oVal} : \oVal[1] \longrightarrow \oVal[2]$ and $\VoFunctor^{\oCmp} : \oCmp[1] \longrightarrow \oCmp[2]$ are cartesian preoperad functors and $\VoFunctor^{\oCmp} \oReturn[1] \equiv \oReturn[2] \VoFunctor^{\oVal}$.
\end{defn}

\subsection{Substitutions and the free construction}\label{subsec:subs}
We now give the main construction of the paper: given a preoperad $\oCmp$, we build a \emph{substitution category} $\oSUBSMB[\oCmp]$ whose morphisms are finite sequences of single substitutions and renamings, modulo the equations of the preoperad axioms. 
Composition is given by concatenation. This yields a pre-PROP $\oSUBSMB[\oCmp]$. 
When $V$ is a cartesian operad, the same construction yields a cartesian PROP $\oCSUBSMB[\oVal]$. 
Moreover, any Freyd operad $\left(\oVal,\oCmp,\oReturn\right)$ induces a functor $\oReturn[\oSUBSMB] : \oCSUBSMB[\oVal] \to \oSUBSMB[\oCmp]$, and therefore a Freyd PROP $\left(\oCSUBSMB[\oVal],\oSUBSMB[\oCmp],\oReturn[\oSUBSMB]\right)$.

In the theory of multicategories, there is a standard adjunction between multicategories and monoidal categories~\cite{hermida2000representable}, and between operads and PROPs~\cite{leinster2004higher}.
The construction of the left adjoint freely adds a tensor product to a multicategory by taking lists of objects for objects, and lists of multicategory morphisms for category morphisms. The resulting monoidal category can be considered as a category of simultaneous substitutions over the multicategory.
As pointed out in~\cite{rajesh2025universal}, this technique cannot work as-is in the case of premulticategories, since lists of morphisms do not encode the sequential order between them.
While~\cite{rajesh2025universal} specifies the appropriate forgetful functor, relating premonoidal categories with premulticategories,  the explicit construction of its left adjoint is left as an open problem. 
While there is probably a way to resolve it meta-categorically, in a similar vein to~\cite{hermida2000representable}, we use a more direct approach, based on the following core principle:
Multicategories are representable presheaves over their category of substitutions.
With that in mind, we construct the category of substitutions by taking the substitutions and renamings used to modify morphisms in a preoperad, and rather than immediately applying them to a morphism in the preoperad, we put them in a list, while omitting the morphism.

\begin{example}\label{ex:sub-word}
Let $\oCmp$ be a preoperad,  $f\in \oCmp(m_{1}+1+m_{2})$,  $g\in \oCmp(n_{1}+1+n_{2})$, and  $h\in \oCmp(n')$.
Consider the iterated action on $f$ obtained by first substituting $g$ into the distinguished input of $f$,
then substituting $h$ into the distinguished input of $g$, and finally reindexing by a renaming $r$: $
\oRenAp{\oSubAp{\oSubAp{f}{m_{1}\oSL g \oSR m_{2}}}{m_{1}+n_{1}\oSL h \oSR n_{2}+m_{2}}}{r}$. 
Instead of applying these operations  to $f$, we record them as the word $
\VoSub := \lPar{\,\oSub{m_{1}\oSL g \oSR m_{2}},\;
               \oSub{m_{1}+n_{1}\oSL h \oSR n_{2}+m_{2}},\;
               r\,}$, 
which remembers the \emph{sequential order} of the two substitutions and the renaming.
In the quotient defining substitutions, $\VoSub$ is identified with $
\lPar{\,\oSub{m_{1}\oSL \oSubAp{g}{n_{1}\oSL h \oSR n_{2}} \oSR m_{2}},\; r\,}$ 
by associativity of substitutions and the unit law for renamings.
\end{example}
 Starting from a preoperad $\oCmp$, we first form words consisting of single substitutions and renamings. 
\begin{defn}[Substitutions]\label{def:sub-constructions}\label{def:op-single-sub}\label{def:op-pre-sub}
Let $\oCmp$ be a preoperad.
\begin{enumerate}
    \item A \emph{single substitution} from $m_{1}+n+m_{2}$ to $m_{1}+1+m_{2}$ is a tuple
    $
    (m_{1},\VoCmp,m_{2})$, where $
     \VoCmp \in \oC{n}$,
    written $\oSub{m_{1} \oSL \VoCmp \oSR m_{2}}$.
    We write $\oSSUB[\oCmp]{m}{n}$ for the set of single substitutions $m \to n$.

    \item A \emph{pre-substitution} $\VoSub:m\to n$ is a finite list $   \VoSub=\lPar{\VoSub[1],\ldots,\VoSub[k]}$
    equipped with arities $n=m_0,m_1,\ldots,m_k=m$ such that for each $i=1,\ldots,k$ the element $\VoSub[i]$ is either a single substitution  $\VoSub[i]\in \oSSUB[\oCmp]{m_{i}}{m_{i-1}}$, or a renaming  
$\VoSub[i] \in \oHom{\oRENSMB^{\mathbf{op}}}{m_{i}}{m_{i-1}}$. 
We write $\oPSUB[C]{m}{n}$ for the set of pre-substitutions $m\to n$, and $\lEmpty\in \oPSUB[\oCmp]{n}{n}$ for the empty list.

\end{enumerate}
\end{defn}

Next, we define two notions of substitutions by defining equivalence relations over pre-substitutions, each of which will correspond to one of our particular preoperads of interest, and the quotient by them.

The cartesian rules of Fig.~\ref{fig:sub-rules} carry a side condition selecting a class $\mathcal{P}$ of \emph{pure} operations to which they apply. We use two instantiations: in a cartesian operad every operation is pure, whereas in a Freyd operad only the operations in the image of $\oReturn$ are.

\begin{figure}[t]
\centering
\small
\fbox{\begin{minipage}{0.97\linewidth}
\textbf{Substitution congruence generators.}
Rules are applied \emph{in context}: $\Gamma\,X\,\Gamma'\cong \Gamma\,Y\,\Gamma'$ whenever well typed.
$\Gamma,\Gamma'$ range over pre-substitution words.

\medskip
\textbf{Symmetric rules:}
\[
\begin{array}{@{}l l >{\scriptstyle}l@{}}
\textbf{(U$_1$)} &
\oSub{m_{1}\oSL \oId \oSR m_{2}} \cong \lEmpty
&(\oId\in \oC{1})
\\[0.35em]
\textbf{(U$_2$)} &
\oRen{\oId[m]}\cong \lEmpty
&(\oId[m]\in\oREN{m}{m})
\\[0.35em]
\textbf{(A)} &
\oSub{m_{1}\oSL g \oSR m_{2}}\;
\oSub{m_{1}+n_{1}\oSL h \oSR n_{2}+m_{2}}
\cong
\oSub{m_{1}\oSL (g\{n_{1}\dashv h \vdash n_{2}\}) \oSR m_{2}}

& (g\in \oC{n_{1}+1+n_{2}},\; h\in \oC{n'})
\\[0.35em]
\textbf{(R)} &

\oRen{r_{1}}\;\oRen{r_{2}}\cong \oRen{r_{2}\circ r_{1}}
&(r_{1}\in\oREN{m}{n},\; r_{2}\in\oREN{n}{k})
\\[0.35em]
\textbf{(N)} &
\oSub{n_{1}\oSL v \oSR n_{2}}\;\oRen{r_{1}+r+r_{2}}
\cong
\oRen{r_{1}+\oId[1]+r_{2}}\;\oSub{m_{1}\oSL v[r] \oSR m_{2}}
&\parbox[t]{0.2\linewidth}{$\scriptstyle
(v\in \oC{m},\; r\in\oREN{m}{n},\\ r_{i}\in\oREN{n_{i}}{m_{i}})
$}
\\[0.35em]
\textbf{(S$_\ell$)} &
\oRen{\oId[m_{1}]+\oSym[1,1]+\oId[m_{2}]}\;\oSub{m_{1}\oSL g \oSR 1+m_{2}}
\cong
\oSub{m_{1}+1\oSL g \oSR m_{2}}\;\oRen{\oId[m_{1}]+\oSym[1,n]+\oId[m_{2}]}
&(g\in \oC{n})
\\[0.35em]
\textbf{(S$_r$)} &
\oRen{\oId[m_{1}]+\oSym[1,1]+\oId[m_{2}]}\;\oSub{m_{1}+1\oSL g \oSR m_{2}}
\cong
\oSub{m_{1}\oSL g \oSR 1+m_{2}}\;\oRen{\oId[m_{1}]+\oSym[n,1]+\oId[m_{2}]}
&(g\in \oC{n})
\end{array}
\]

\smallskip
\textbf{Additional cartesian rules} (for a class $\mathcal{P}$ of \emph{pure} operations; see Def.~\ref{def:op-sub})\textbf{:}
\[
\begin{array}{@{}l l >{\scriptstyle}l@{}}
\textbf{(CEN)} &
\oSub{m_{1}\oSL g_{1} \oSR m+1+m_{2}}\;
\oSub{m_{1}+n_{1}+m\oSL g_{2} \oSR m_{2}}
\cong
\oSub{m_{1}+1+m\oSL g_{2} \oSR m_{2}}\;
\oSub{m_{1}\oSL g_{1} \oSR m+n_{2}+m_{2}}
&
\parbox[t]{0.09\linewidth}{$\scriptstyle
(g_{1}\in \mathcal{P}(n_{1}),\\ g_{2}\in \mathcal{P}(n_{2}))
$}
\\[0.35em]
\textbf{(D)}
&
\oRen{\oId[m_{1}]+\oDiscard[1]+\oId[m_{2}]}\;\oSub{m_{1}\oSL g \oSR m_{2}}
\cong
\oRen{\oId[m_{1}]+\oDiscard[n]+\oId[m_{2}]}
&(g\in \mathcal{P}(n))
\\[0.35em]
\textbf{(C)} &
\oRen{\oId[m_{1}]+\oCopy[1]+\oId[m_{2}]}\;\oSub{m_{1}\oSL g \oSR m_{2}}
\cong
\oSub{m_{1}\oSL g \oSR 1+m_{2}}\;\oSub{m_{1}+n\oSL g \oSR m_{2}}\;\oRen{\oId[m_{1}]+\oCopy[n]+\oId[m_{2}]}
&(g\in \mathcal{P}(n))
\end{array}
\]
\end{minipage}%
}
\vspace{-0.3em}
\caption{Rules for the equivalence relation on pre-substitutions.}
\label{fig:sub-rules}
\end{figure}

\begin{defn}[Symmetric and Cartesian Substitutions]\label{def:op-sub}
When $\oCmp$ is a symmetric $\oRENSMB$-cartesian preoperad,
$
    \oSUB[\oCmp]{m}{n}$
    is the quotient of $\oPSUB[\oCmp]{m}{n}$ by the symmetric rules of Fig.~\ref{fig:sub-rules}.
    When $\oVal$ is a cartesian operad,
$
    \oCSUB[\oVal]{m}{n}$
    is the quotient of $\oPSUB[\oVal]{m}{n}$ by all rules of Fig.~\ref{fig:sub-rules}, taking $\mathcal{P}(n) := \oV{n}$.
    When $\left( \oVal , \oCmp , \oReturn \right)$ is a Freyd operad,
$
    \oSUB[\oCmp]{m}{n}$
    is the quotient of $\oPSUB[\oCmp]{m}{n}$ by the symmetric rules \emph{together with} the cartesian rules for $\mathcal{P}(n) := \{\, \oReturn \VoVal \mid \VoVal \in \oV{n} \,\}$, the image of $\oReturn$.
\end{defn}

The congruence of Fig.~\ref{fig:sub-rules} is designed so that the resulting quotients inherit the standard
PROP structure by acting componentwise on words: composition is concatenation, and whiskering extends
each step by unused variables on the left or right.

\begin{defn}[Freyd PROP of Substitutions]\label{def:freyd-sub}
Given a Freyd operad $\left( \oVal , \oCmp, \oReturn\right)$, define an identity-on-objects map $\oReturn^{\oSUBSMB} : \oCSUBSMB[\oVal] \longrightarrow \oSUBSMB[\oCmp]$ by recursion on substitution words:
\[
\begin{array}{lcl}
    \oReturn^{\oSUBSMB} \lEmpty & := & \lEmpty \\
    \oReturn^{\oSUBSMB} \lPar{ \oRen{\VoRen}, \VoSub}& := & \lPar{\oRen{\VoRen} , \oReturn^{\oSUBSMB} \VoSub}\\
    \oReturn^{\oSUBSMB} \lPar{ \oSub{m_{1} \oSL \VoVal \oSR m_{2}} , \VoSub}& := & \lPar{\oSub{m_{1} \oSL \oReturn \VoVal \oSR m_{2}} , \oReturn^{\oSUBSMB} \VoSub}
\end{array}
\]
The functoriality of $\oReturn$, together with the defining equations of $\oVal$ and $\oCmp$, ensures that $\oReturn^{\oSUBSMB}$ preserves the equivalence relation.
\end{defn}

\begin{proposition}\label{prop:sub-constructions}
The following hold.
\begin{enumerate}
    \item If $\oCmp$ is a symmetric $\oRENSMB$-cartesian preoperad, then $\oSUBSMB[\oCmp]$ is a pre-PROP.

    \item If $\oVal$ is a cartesian operad, then $\oCSUBSMB[\oVal]$ is a cartesian PROP.

    \item If $(\oVal,\oCmp,\oReturn)$ is a Freyd operad, then
$\bigl(\oCSUBSMB[\oVal],\oSUBSMB[\oCmp],\oReturn^{\oSUBSMB}\bigr)$ is a Freyd PROP.

    \item Given a Freyd operad functor
$    (\VoFunctor^{\oVal},\VoFunctor^{\oCmp}) :
    (\oVal[1],\oCmp[1],\oReturn[1])
    \longrightarrow
    (\oVal[2],\oCmp[2],\oReturn[2])$,
    there is an induced Freyd PROP functor $
    (\VoFunctor^{\oCSUBSMB},\VoFunctor^{\oSUBSMB}) :
\bigl(\oCSUBSMB[{\oVal[1]}],\oSUBSMB[{\oCmp[1]}],\oReturn[1]^{\oSUBSMB}\bigr)
    \longrightarrow
    \bigl(\oCSUBSMB[{\oVal[2]}],\oSUBSMB[{\oCmp[2]}],\oReturn[2]^{\oSUBSMB}\bigr)$, 
    obtained by applying $\VoFunctor^{\oVal}$ and $\VoFunctor^{\oCmp}$ to each substitution step and leaving renamings unchanged.
\end{enumerate}
\end{proposition}
\begin{proof}~
\begin{enumerate}
    \item The rules of Fig.~\ref{fig:sub-rules} define a congruence on pre-substitutions, so concatenation descends to a well-defined associative composition with unit $\lEmpty$. Whiskering is defined componentwise on steps, and each generating rule is stable under whiskering, so tensor is well-defined on the quotient. The symmetry is represented by the singleton renaming word $[\sigma_{m,n}]$, and its naturality follows from the renaming/substitution interaction and the symmetry rules.

\item The additional cartesian rules force centrality of substitution steps and the naturality of copying and discarding. The copy and discard maps are represented by the singleton renaming words $[\Delta_{n}]$ and $[!_{n}]$, and the comonoid laws,  together with the coherence conditions of Def.~\ref{def:cart-PROP}, are inherited from $\oRENSMB^{\mathbf{op}}$.

\item The map $\oReturn^{\oSUBSMB}$ is well defined because $\oReturn$ preserves renamings and the generating equations. Its image is central in $\oSUBSMB[\oCmp]$: a word in the image is a concatenation of renaming steps, which are central by (i), and of substitution steps $\oSub{m_{1} \oSL \oReturn \VoVal \oSR m_{2}}$, which commute with every other substitution step by rule \textbf{(CEN)} for $\mathcal{P}$, the image of $\oReturn$ (Def.~\ref{def:op-sub}) and past renaming steps by \textbf{(N)} and \textbf{(R)}. Hence $\oReturn^{\oSUBSMB}$ defines a Freyd PROP functor.

\item The induced maps lists respect the quotients since $\VoFunctor^{\oVal}$ and $\VoFunctor^{\oCmp}$ preserve the relevant operadic structure and renamings. Compatibility with $\oReturn^{\oSUBSMB}$ follows from the Freyd-operad functor axioms.
\qed
\end{enumerate}
\let\qed\relax
\end{proof}

\subsection{Representability}\label{subsec:presheaves}

The substitution categories of Sec.~\ref{subsec:subs} are designed so that arity families of a preoperad can be
recovered as morphisms with codomain $1$.
Concretely, a word $s:m\to n$ in $\oSUBSMB[\oCmp]$ should act on any $n$-ary operation $f\in \oC{n}$ by producing an
$m$-ary operation obtained by performing the indicated substitutions and renamings in sequence.
This yields a contravariant action $\oC{n}\to \oC{m}$, hence a presheaf $\oCmp:{\oSUBSMB[\oCmp]}^{\mathrm{op}}\to \mathbf{Set}$, and we show that this presheaf is representable by the implicit generating object $1$.

\begin{defn}[Pre-Substitution Application]
Let $\oCmp$ be a symmetric $\oRENSMB$-cartesian preoperad, 
$\VoSub \in \oPSUB[\oCmp]{m}{n}$, and $\VoCmpI \in \oC{n}$. Define the \emph{application} $\oPshAp{\VoCmpI}{\VoSub}\in \oC{m}$, by recursion on $\VoSub$:
$$
    \oPshAp{\VoCmpI}{\lEmpty}  := \VoCmp \qquad,\qquad
    \oPshAp{\VoCmpI}{\lPar{ \oSub{m_{1} \oSL \VoCmpII \oSR m_{2}} , \VoSub}}  := \oPshAp{\oSubAp{\VoCmpI}{m_{1} \oSL \VoCmpII \oSR m_{2}}}{\VoSub} \qquad,\qquad
    \oPshAp{\VoCmpI}{\lPar{ \VoRen , \VoSub}}  := \oPshAp{\oRenAp{\VoCmpI}{\VoRen}}{\VoSub}
$$
\end{defn}

\begin{theorem}\label{thm:pre-sub-equiv}
The following hold.
\begin{enumerate}
\item If $\oCmp$ is a symmetric $\oRENSMB$-cartesian preoperad, and $\VoSub[1] \cong \VoSub[2] \in \oSUB[\oCmp]{m}{n}$ for the quotient by the symmetric rules of Fig.~\ref{fig:sub-rules}, then for all $\VoCmp \in \oC{n}$: 
$\oPshAp{\VoCmp}{\VoSub[1]} \equiv \oPshAp{\VoCmp}{\VoSub[2]} \in \oC{m} $.

    \item If $\oVal$ is a cartesian operad, and $\VoSub[1] \cong \VoSub[2] \in \oCSUB[\oVal]{m}{n}$ then for all $\VoVal \in \oV{n}$: 
$\oPshAp{\VoVal}{\VoSub[1]} \equiv \oPshAp{\VoVal}{\VoSub[2]} \in \oV{m} $.

\item If $\left( \oVal , \oCmp , \oReturn \right)$ is a Freyd operad, 
and $\VoSub[1] \cong \VoSub[2] \in \oSUB[\oCmp]{m}{n}$, then for all $\VoCmp \in \oC{n}$: 
$\oPshAp{\VoCmp}{\VoSub[1]} \equiv \oPshAp{\VoCmp}{\VoSub[2]} \in \oC{m} $.
\end{enumerate}
\end{theorem}
\begin{proof}
For (i), it suffices to check each generating rule of Fig.~\ref{fig:sub-rules}.
Each rule is exactly one of the axioms of a symmetric $\oRENSMB$-cartesian preoperad (units, associativity,
renaming functoriality, renaming-substitution interchange, and symmetry naturality).
For (ii), the additional cartesian rules correspond to centrality and to the naturality axioms for copying and
discarding in a cartesian operad.
For (iii), the symmetric rules are as in (i), and the three cartesian rules, now restricted to substitutions over the image of $\oReturn$, are precisely conditions (i), (ii) and (iii) of Def.~\ref{def:freyd-operad}.
\end{proof}

The implication in Thm.~\ref{thm:pre-sub-equiv} is one-directional and it states that our congruence is \emph{sound}: equivalent words act identically on every operation of $\oCmp$,
so that application descends to the quotient.
The converse need not hold for a fixed preoperad $\oCmp$, since $\oCmp$ may satisfy additional equations beyond
those enforced by the defining axioms, which may identify more words extensionally.
Thus the congruence should be understood as the least one generated by the structural axioms in Fig.~\ref{fig:sub-rules}.

\begin{corollary}\label{cor:sub-presheaf}
For symmetric $\oCmp$, the assignment $n\mapsto \oC{n}$ extends to a functor
$\oCmp:\oSUBSMB[\oCmp]^{\mathrm{op}}\to \mathbf{Set}$ by $\oCmp(s)(f)=f\star s$.
Similarly, for cartesian $\oVal$ we obtain $V:{\oCSUBSMB[\oVal]}^{\mathrm{op}}\to \mathbf{Set}$.
\end{corollary}

\begin{theorem}\label{thm:sub-presheaf}
Let $\left( \oVal , \oCmp, \oReturn \right)$ be a Freyd operad. 
\begin{enumerate}
\item The presheaf $\oVal:{\oCSUBSMB[\oVal]}^{\mathrm{op}}\to\mathbf{Set}$ is representable by $1$, with
$\oV{n}\;\cong\;\oCSUB[\oVal]{n}{1}$.
\item The presheaf $\oCmp:\oSUBSMB[\oCmp]^{\mathrm{op}}\to\mathbf{Set}$ is representable by $1$, with
$\oC{n}\;\cong\;\oSUB[\oCmp]{n}{1}$.
\item The map $\oReturn:\oVal\to \oCmp$ is a natural transformation
$\oVal \Rightarrow \oCmp \oReturn^{\oSUBSMB}$.
\end{enumerate}
\end{theorem}
\begin{proof}
\begin{enumerate}
\item The action of both categories is given by pre-substitution application. Thm.~\ref{thm:pre-sub-equiv} ensures that it is well-defined, as it preserves the equivalence relations defining the appropriate categories.
Functoriality is given by construction, since $\oPshAp{\VoCmp}{\lEmpty} = \VoCmp$ and induction on $\VoSub[1]$ verifies that $\oPshAp{\VoCmp}{\left(\VoSub[1] \lApp \VoSub[2]\right)} \equiv \oPshAp{\left(\oPshAp{\VoCmp}{\VoSub[1]}\right)}{\VoSub[2]}$.
To prove representability for $\oVal$, we go from $\oV{n}$ to $\oCSUB[\oVal]{n}{1}$ by taking each $\VoVal \in \oV{n}$ to $\lPar{\oSub{0 \oSL \VoVal \oSR 0}} \in \oCSUB[\oVal]{n}{1}$, and in the other direction by taking $\VoSub \in \oCSUB[\oVal]{n}{1}$ to $\oPshAp{\oId}{\VoSub} \in \oV{n}$. From $\oPshAp{\oId}{\lPar{\oSub{0 \oSL \VoVal \oSR 0}}} \equiv \VoVal$ we get one part of the isomorphism.
For the other part, we prove by induction the following naturality property:
For every $k$, every $\VoVal \in \oV{k}$ and every $\VoSub \in \oCSUB[\oVal]{n}{k}$, 
$
\lPar{\oSub{0 \oSL \oPshAp{\VoVal}{\VoSub} \oSR 0}} \;\cong\; \lPar{\oSub{0 \oSL \VoVal \oSR 0}} \lApp \VoSub
$.
From this, we get $\lPar{\oSub{0 \oSL \oPshAp{\oId}{\VoSub} \oSR 0}} \cong \VoSub$ by setting $\VoVal := \oId$.

\item The same method is used to prove representability for $\oCmp$.

\item For naturality of $\oReturn$, we need to prove that for every $m,n\in\N$,  $\VoSub \in \oCSUB[\oVal]{m}{n}$, and  $\VoVal \in \oV{n}$, we get $\oPshAp{\oReturn \VoVal}{\oReturn^{\oSUBSMB} \VoSub} \equiv \oReturn\left( \oPshAp{\VoVal}{\VoSub}\right)$. This is shown by induction on $\VoSub$, using the functoriality of $\oReturn$.
\qed
\end{enumerate}\let\qed\relax
\end{proof}

Representability yields natural bijections
$
\oC{n}\;\cong\;\oSUB[\oCmp]{n}{1}$ 
 and 
$\oV{n}\;\cong\;\oCSUB[\oVal]{n}{1}$. 
Thus a substitution PROP recovers its underlying (Freyd) operad by restriction to codomain $1$. 

\subsection{The Freyd adjunction}\label{subsec:adjunction}

We now prove the universal property of the substitution construction.
The  representability result (Thm.~\ref{thm:sub-presheaf}) packages restriction to codomain $1$ as a functor 
 $U : \oFREYDPROP \to \oFreydOp$, and we show that the free substitution construction $F : \oFreydOp \to \oFREYDPROP$ is left adjoint to $U$.
We use $\oFREYDPROP$ to denote the category of (small) Freyd PROP categories, with Freyd PROP functors between them.
We use $\oFreydOp$ to denote the category of (small) Freyd operads, with Freyd operad functors between them.

\begin{defn}[The Forgetful Functor $\oForget$]\label{def:forget}
Define the functor 
$\oForget : \oFREYDPROP \to \oFreydOp$ as follows.

On objects, $\oForget$ sends a Freyd PROP $(\oVal,\oCmp,\oReturn)$ to the Freyd operad
$(\oVal^{\oForget},\oCmp^{\oForget},\oReturn^{\oForget})$ where\vspace{-0.3em}
\[
\oVal^{\oForget}(n) := \oHom{\oVal}{n}{1}
\qquad\qquad
\oCmp^{\oForget}(n) := \oHom{\oCmp}{n}{1}
\]
and $\oReturn^{\oForget}$ is the restriction of $\oReturn$ to morphisms with codomain $1$.
The operadic identity is $\oId[1]\in \oHom{\oVal}{1}{1}$ (and similarly in $\oCmp$), and the operadic substitution is induced by
whiskering and composition in the PROP:
for $f\in \oHom{\oVal}{n_1+1+n_2}{1}$ and $g\in \oHom{\oVal}{k}{1}$,
$
\oSubAp{f}{n_1 \oSL g \oSR n_2} \;:=\; f \circ (n_1 \oSL g \oSR n_2)
$, 
and similarly in $\oCmp$.
Renamings in $\oVal^{\oForget}$ are interpreted in $\oVal$ via Lemma~\ref{lemma:ren-gen} and the $\oRENSMB$-cartesian structure is given by:
$ \oRenAp{\VoVal}{ \VoRen } := \VoVal \circ \VoRen  $.
In $\oCmp^{\oForget}$, renamings are the same renamings as in $\oVal^{\oForget}$, but with:
$\oRenAp{\VoCmp}{ \VoRen } := \VoCmp \circ \oReturn \VoRen  $.

On morphisms, $\oForget$ sends a Freyd PROP functor $(\VoFunctor^{\oVal},\VoFunctor^{\oCmp})$ to the Freyd operad functor
$(\oForget\VoFunctor^{\oVal},\oForget\VoFunctor^{\oCmp})$ obtained by restricting $\VoFunctor^{\oVal}$ and $\VoFunctor^{\oCmp}$
to morphisms with codomain $1$.
\end{defn}

\begin{defn}[The Free Functor $\oFree$]\label{def:free}
Define the functor 
$\oFree : \oFreydOp \to \oFREYDPROP$ as follows.

On objects, $\oFree$ sends a Freyd operad $(\oVal,\oCmp,\oReturn)$ to the Freyd PROP
$
\oFree(\oVal,\oCmp,\oReturn) \;:=\; (\oCSUBSMB[\oVal],\,\oSUBSMB[\oCmp],\,\oReturn^{\oSUBSMB})
$, 
as in Def.~\ref{def:freyd-sub}.
On morphisms, $\oFree$ sends a Freyd operad functor $(\VoFunctor^{\oVal},\VoFunctor^{\oCmp})$ to the induced Freyd PROP functor
$(\VoFunctor^{\oCSUBSMB[\oVal]},\VoFunctor^{\oSUBSMB[\oCmp]})$ from Prop.~\ref{prop:sub-constructions}.
\end{defn}

\begin{theorem}[The Freyd Adjunction]\label{thm:preop-adj}
There is an adjunction $\oFree \dashv \oForget$ where\vspace{-0.3em}
\[
\oForget : \oFREYDPROP \longrightarrow \oFreydOp
\qquad\text{and}\qquad
\oFree : \oFreydOp \longrightarrow \oFREYDPROP.
\]
\end{theorem}
\begin{proof}
Given a Freyd operad $\VoOperad$ and Freyd PROP $\VoProp$,
we define the following isomorphism between $\oHom{\oFREYDPROP}{\oFree(\VoOperad)}{\VoProp}$ and $\oHom{\oFreydOp}{\VoOperad}{\oForget(\VoProp)}$.
One direction is obtained by restricting a Freyd PROP functor to single substitutions, and the other by interpreting substitution lists in the target Freyd PROP. These constructions are mutually inverse and natural in both variables, hence determine a natural bijection of hom-sets. Therefore $\oFree \dashv \oForget$.
\end{proof}

\section{Interpretation in Freyd operads}\label{sec:semantics}

In this section we interpret $\xCLC$ in Freyd operads, prove soundness, and construct the canonical term model. 
A context of length $n$ is interpreted as the arity $n$, values as morphisms in $\oV{n}$, and computations as morphisms in $\oC{n}$. 
The first-order fragment is interpreted in any Freyd operad, but to interpret abstraction and application we additionally require a weak closure structure. This is the untyped, intensional analogue of the usual closed structure used for higher-order categorical semantics, presented through a representation of an appropriate left module following~\cite{staton2013universal}.
Our weak closure validates beta reduction but does not impose eta laws. To model eta, it suffices to add the axiom $\oAbs[n]{\oSubAp{\oApp}{\oReturn \VoVal \oSR 1}} = \VoVal$, recovering the standard notion of function space from~\cite{staton2013universal}. Since we work in an untyped setting, there is no explicit representing object, but the construction is otherwise the same.

\subsection{Models}\label{subsec:models}

We first define the semantic notion of model for $\xCLC$.

\begin{defn}[Weakly Closed Freyd Operad]
\label{def:weakly-closed}
A Freyd operad $\left( \oVal , \oCmp , \oReturn \right)$ is \emph{weakly closed} if it is equipped with
a family of maps $\oAbs[n]{-}:\oC{n+1} \rightarrow \oV{n}$ of \emph{abstraction}, for $n \in \N$, and a morphism $\oApp \in \oC{2}$ of \emph{application}, such that the following hold  (whenever they are defined):
$$
\oSubAp{\oApp}{\oReturn \oAbs[n]{\VoCmp} \oSR 1}  =  \VoCmp
    \qquad \qquad
    \oAbs[m_{1} + n + m_{2}]{ \oSubAp{ \VoCmp }{ m_{1} \oSL \oReturn \VoVal \oSR m_{2} + 1 } }
     =  \oSubAp{\oAbs[m_{1} + 1 + m_{2}]{ \VoCmp }}{m_{1} \oSL \VoVal \oSR m_{2}}
$$
and, for every $\VoCmp \in \oC{m+1}$ and every $\VoRen \in \oREN{m}{n}$, abstraction is natural in the surrounding context:
$$
\oAbs[n]{ \oRenAp{\VoCmp}{\VoRen + \oId[1]} } \;=\; \oRenAp{\oAbs[m]{\VoCmp}}{\VoRen} .
$$
\end{defn}
Thus abstraction is a weak representation of the left $\oVal$-module $\oC{n+1}$, with $\oApp$ providing the corresponding counit, natural in the arity.
\begin{defn}[Closed Freyd Operad Functor]
A Freyd operad functor between two weakly closed Freyd operads $\VoFunctor : \left(\oVal[1], \oCmp[1] , \oReturn[1] \right) \longrightarrow \left(\oVal[2], \oCmp[2] , \oReturn[2] \right)$
is \emph{closed} when it preserves abstraction and application: $\VoFunctor\oApp = \oApp$, and $\VoFunctor \oAbs[n]{\VcCmp} =  \oAbs[n]{\VoFunctor \VcCmp}$.
\end{defn}

\begin{defn}[$\xCLC$-Structure]\label{def:structure}
Given a signature $\xSign = \left( \xFuncSMB , \xProcSMB \right)$, a $\xCLC$-structure for $\xSign$ is a tuple $\VoStruct = \left( \oVal , \oCmp , \oReturn , \oFuncAsgSMB , \oProcAsgSMB \right)$ where $\left( \oVal , \oCmp , \oReturn\right)$ is a weakly closed Freyd operad, 
$\oFuncAsgSMB$ assigns to each $\VxFunc \in \xFuncSymb{n}$ a morphism $\oFuncAsg{\VxFunc} \in \oV{n}$, and $\oProcAsgSMB$ assigns to each $\VxProc \in \xProcSymb{n}$ a morphism $\oProcAsg{\VxProc} \in \oC{n}$.
\label{def:mor-structure} 
A morphism $\VoFunctor: \VoStruct[1] \to \VoStruct[2]$ of $\xCLC$-structures is a closed Freyd operad functor between the underlying weakly closed Freyd operads that also preserves the symbol interpretations: 
$
\VoFunctor\bigl(\oFuncAsg{\VxFunc}\bigr) = \oFuncAsg{\VxFunc} $ 
and 
$
\VoFunctor\bigl(\oProcAsg{\VxProc}\bigr) = \oProcAsg{\VxProc} $.
\end{defn}

We will later construct a canonical \emph{term $\xCLC$-structure} from $\xCLC$ terms modulo the equational theory and show that it is initial.

\subsection{Interpretation and Soundness}

Fix a signature $\xSign = \left( \xFuncSMB , \xProcSMB \right)$ and $\xCLC$-structure $\VoStruct=(\oVal,\oCmp,\oReturn,\oFuncAsgSMB,\oProcAsgSMB)$ over $\xSign$. We write $\lLength{\VxCxt}$ for the length of a context, and we freely use concatenation notation $\VxCxt=\VxCxt[1],\ldots,\VxCxt[k]$ when a term is formed from subterms in the corresponding subcontexts.
The interpretation follows the value/computation split of the syntax: variables and function symbols are interpreted in the cartesian value operad,  computations in the computation preoperad, the let-constructor by single substitution in $\oCmp$, and abstraction and application by the weak closure structure.

\begin{defn}[Interpretation]\label{def:interp-freyd}
Given a $\xCLC$-structure $\VoStruct$, we define two families of interpretation functions, $\oInterpVal[\VxCxt]{-} : \xVAL{\VxCxt} \rightarrow \oV{\lLength{\VxCxt}}$ and $\oInterpCmp[\VxCxt]{-} : \xCMP{\VxCxt} \rightarrow \oC{\lLength{\VxCxt}}$, by mutual recursion:
$$
\renewcommand{\arraystretch}{1.1}
\setlength{\arraycolsep}{10pt}
\begin{array}{lcl}
     \oInterpVal[{\VxVar[1] , \ldots , \VxVar[n]}]{\VxVar[k]} &:=& \oRenAp{\oId}{\oDiscard[k-1] + \oId[1] + \oDiscard[n-k]}\\
     \oInterpVal[{\VxCxt[1] , \ldots , \VxCxt[k]}]{\VxFunc \left(\VxVal[1], \ldots , \VxVal[k]\right)} &:=& \oSubAp{ \oFuncAsg{\VxFunc} }{ \oInterpVal[{\VxCxt[1]}]{\VxVal[1]} + \ldots + \oInterpVal[{\VxCxt[k]}]{\VxVal[k]} } \\
         \oInterpVal[\VxCxt]{ \TxAbs{\VxVar}{\VxCmp} } & := & \oAbs[\lLength{\VxCxt}]{ \oInterpCmp[\VxCxt , \VxVar]{\VxCmp} }
    \\
     \oInterpCmp[\VxCxt]{\TxRet{\VxVal}} &:=& \oReturn \oInterpVal[\VxCxt]{\VxVal}\\
     \oInterpCmp[{\VxCxt[1] , \VxCxt[2]}]{\TxLet{\VxVar}{\VxCmp[2]}{\VxCmp[1]}} &:=& \oSubAp{\oInterpCmp[{\VxCxt[1] , \VxVar}]{\VxCmp[1]}}{ \lLength{\VxCxt[1]} \, \oSL \oInterpCmp[{\VxCxt[2]}]{ \VxCmp[2]} } \\
     \oInterpCmp[{\VxCxt[1] , \ldots , \VxCxt[k]}]{\VxProc \left(\VxVal[1], \ldots , \VxVal[k]\right)} &:=& \oSubAp{ \oProcAsg{\VxProc} }{ \oReturn \oInterpVal[{\VxCxt[1]}]{\VxVal[1]} + \ldots + \oReturn \oInterpVal[{\VxCxt[k]}]{\VxVal[k]} }\\
    \oInterpCmp[{\VxCxt[1],\VxCxt[2]}]{\TxApp{\VxVal[1]}{\VxVal[2]}} & := & \oSubAp{\oApp}{ \oReturn \oInterpVal[{\VxCxt[1]}]{\VxVal[1]} + \oReturn \oInterpVal[{\VxCxt[2]}]{\VxVal[2]} }
\end{array}
$$
\end{defn}

We next prove the fundamental lemma, namely, that interpretation preserves substitution. 

\begin{lemma}\label{lemma:interp-sub}
Given a structure $\VoStruct$, for any contexts $\VxCxt[1] , \ldots , \VxCxt[n]$ and values $\VxVal[1] , \ldots , \VxVal[n]$ s.t  $\JxWFV{\VxCxt[i]}{\VxVal[i]}$, the following hold for any  $\VxVal$ s.t $\JxWFV{\VxVar[1] , \ldots , \VxVar[n]}{\VxVal}$ and any  $\VxCmp$ s.t $\JxWFC{\VxVar[1] , \ldots , \VxVar[n]}{\VxCmp}$:\vspace{-0.3em}
$$
\begin{array}{rcl}
    \oInterpVal[{\VxCxt[1] , \ldots , \VxCxt[n]}]{ \VxVal \subst{\assign{\VxVar[1]}{\VxVal[1]} , \ldots , \assign{\VxVar[n]}{\VxVal[n]}}}
    &\equiv&
    \oSubAp{\oInterpVal[{\VxVar[1] , \ldots , \VxVar[n]}]{\VxVal}}{\oInterpVal[{\VxCxt[1]}]{\VxVal[1]} + \ldots +  \oInterpVal[{\VxCxt[n]}]{\VxVal[n]} }\\
    \oInterpCmp[{\VxCxt[1] , \ldots , \VxCxt[n]}]{\VxCmp  \subst{\assign{\VxVar[1]}{\VxVal[1]} , \ldots , \assign{\VxVar[n]}{\VxVal[n]}}}
    &\equiv&
    \oSubAp{\oInterpCmp[{\VxVar[1] , \ldots , \VxVar[n]}]{\VxCmp}}{\oReturn \oInterpVal[{\VxCxt[1]}]{\VxVal[1]} + \ldots + \oReturn \oInterpVal[{\VxCxt[n]}]{\VxVal[n]} }\\
\end{array}
$$
\end{lemma}

\begin{defn}[Satisfiability]
Let $\VoStruct$ be a $\xCLC$-structure. 
For well formed terms $\VxTrm[1],\VxTrm[2]$ of the same sort in a context $\VxCxt$, we write
$\JxSat{\VoStruct}{\VxCxt}{\VxTrm[1]}{\VxTrm[2]}$
when their interpretations are equal in the corresponding semantic sort, i.e., in $\oV{\lLength{\VxCxt}}$ for values and in $\oC{\lLength{\VxCxt}}$ for computations. 
We write
$\JxSatVal{\VoStruct}{\VxCxt}{\VxVal[1]}{\VxVal[2]}$ 
and 
$\JxSatCmp{\VoStruct}{\VxCxt}{\VxCmp[1]}{\VxCmp[2]}$
for value and computation, respectively.
We write
$\JxValid{\VxCxt}{\VxTrm[1]}{\VxTrm[2]}$
when $\JxSat{\VoStruct}{\VxCxt}{\VxTrm[1]}{\VxTrm[2]}$ holds for every $\xCLC$-structure $\VoStruct$, and
$\JxValidVal{\VxCxt}{\VxVal[1]}{\VxVal[2]}$ and 
$\JxValidCmp{\VxCxt}{\VxCmp[1]}{\VxCmp[2]}$
for the two sorts, respectively.
\end{defn}

We now show that the equational theory of $\xCLC$ is respected by the interpretation. 
Intuitively, each equation in Fig.~\ref{fig:clc-theory} expresses one of the defining laws of a (weakly closed) Freyd operad: the sequencing rules correspond to the preoperad substitution laws in $\oCmp$, and beta corresponds to weak closure. Well formedness assumptions play no semantic role beyond ensuring that the relevant interpretations are defined.

\begin{theorem}[Soundness]\label{thm:soundness}
Given well-formed terms $\VxTrm[1] , \VxTrm[2]$ in $\VxCxt$:
 $\JxEq{\VxCxt}{\VxTrm[1]}{\VxTrm[2]} \; \Longrightarrow \; \JxValid{\VxCxt}{\VxTrm[1]}{\VxTrm[2]} $.
\end{theorem}
\begin{proof}
By induction on the derivation of $\JxEq{\VxCxt}{\VxTrm[1]}{\VxTrm[2]}$. The structural cases are immediate, and $\beta$ follows from weak closure. For \emph{lunit}, the interpretation of $\TxLet{\VxVar}{\TxRet{\VxVal}}{\VxCmp}$ is $\oSubAp{\oInterpCmp[{\VxCxt[1],\VxVar}]{\VxCmp}}{\lLength{\VxCxt[1]}\oSL \oReturn\oInterpVal[{\VxCxt[2]}]{\VxVal}}$, which is the interpretation of the right hand side by Lemma~\ref{lemma:interp-sub}. The \emph{runit} case is exactly the left unit law of the preoperad, and \emph{assoc} is exactly the associativity law of substitution in the preoperad.
\end{proof}

\subsection{Term model and Initiality}

We now construct the canonical \emph{term model} of $\xCLC$. It packages the syntax modulo the equational theory of Fig.~\ref{fig:clc-theory} into a weakly closed Freyd operad, and it enjoys the expected universal property. As usual, initiality yields completeness as an immediate corollary.

For each $n \in \N$, let $\xVAL{n}$ be the set of equivalence classes $[\VxVal]$ of well formed values $\JxWFV{\VxVar[1],\ldots,\VxVar[n]}{\VxVal}$, modulo provable value equality in Fig.~\ref{fig:clc-theory}. Similarly, let $\xCMP{n}$ be the set of equivalence classes $[\VxCmp]$ of well formed computations $\JxWFC{\VxVar[1],\ldots,\VxVar[n]}{\VxCmp}$, modulo provable computation equality. We write $\xVALSMB := (\xVAL{n})_{n\in\N}$ and $\xCMPSMB := (\xCMP{n})_{n\in\N}$.

The operadic structures are induced on representatives: for values, $\oSubAp{[\VxVal[1]]}{n_{1}\oSL[\VxVal[2]]\oSR n_{2}} := [\VxVal[1]\subst{\assign{\VxVarII}{\VxVal[2]}}]$, and for computations, $\oSubAp{[\VxCmp[1]]}{n_{1}\oSL[\VxCmp[2]]\oSR n_{2}} := [\TxLet{\VxVarII}{\VxCmp[2]}{\VxCmp[1]}]$, where $\VxVarII$ is the distinguished variable in the middle position. 
Renamings are induced by syntactic renaming on representatives. 
These operations are well defined by the congruence rules and the substitution laws of the syntax. 
The return map is induced by the constructor $\TxRet{-}$, namely $\TxRetSMB([\VxVal]) := [\TxRet{\VxVal}]$. 
The interpretations of primitive symbols are given by their canonical terms:
$\oFuncAsg{\VxFunc} := [\VxFunc(\VxVar[1],\ldots,\VxVar[n])] \in \xVAL{n}$
and
$\oProcAsg{\VxProc} := [\VxProc(\VxVar[1],\ldots,\VxVar[n])] \in \xCMP{n}$.
Finally, weak closure is induced by the term formers for application and abstraction: $\oApp := [\TxApp{\VxVar[1]}{\VxVar[2]}] \in \xCMP{2}$ and, for $\JxWFC{\VxVar[1],\ldots,\VxVar[n],\VxVar}{\VxCmp}$, $\oAbs[n]{[\VxCmp]} := [\TxAbs{\VxVar}{\VxCmp}] \in \xVAL{n}$.

\begin{theorem}\label{thm:free-model}
$\left( \xVALSMB , \xCMPSMB , \TxRet{} , \xFuncSMB , \xProcSMB \right)$ is a $\xCLC$-structure.
\end{theorem}
\begin{proof}
The operad and preoperad laws for $\xVALSMB$ and $\xCMPSMB$ follow from the corresponding substitution laws of the syntax, together with the exchange and renaming rules. 
That $\TxRetSMB$ is a Freyd operad functor is witnessed by the $\TxLet{}{}{}$ unit and associativity equations, and centrality and cartesian structure are witnessed by the structural rules in Fig.~\ref{fig:clc-theory}. 
Weak closure holds because the equations of Def.~\ref{def:weakly-closed} are exactly the beta axiom, the compatibility of abstraction with substitution, and the commutation of abstraction with renaming, all three provable in the theory. Finally, the symbol assignments are well formed by construction and satisfy no further axioms.
\end{proof}

Let $\oFreydOp_{\VoStruct}$ be the category whose objects are $\xCLC$-structures over the fixed signature $\xSign$ and whose morphisms are Freyd operad functors preserving the interpretations of function and procedure symbols. 

\begin{theorem}[Initiality]\label{thm:initiality}
The term model $\left( \xVALSMB , \xCMPSMB , \TxRet{} , \xFuncSMB , \xProcSMB \right)$ is  initial  in $\oFreydOp_{\VoStruct}$.
\end{theorem}
\begin{proof}
Let $\left( \oVal , \oCmp , \oReturn , \oFuncAsgSMB, \oProcAsgSMB \right)$ be a $\xCLC$-structure.
First we show that the interpretation function given in Def.~\ref{def:interp-freyd} forms a Freyd operad functor.
Identity is preserved by sending $\JxWFV{\VxVar}{\VxVar}$ to $\oRenAp{\oId}{\oId[1]} \equiv \oId$ in $\oV{1}$ and $\JxWFC{\VxVar}{\TxRet{\VxVar}}$ to $\oReturn\left(\oId\right)  \equiv \oId$ in $\oC{1}$.
Composition of values is preserved by Lem.~\ref{lemma:interp-sub}.
Composition of computations is preserved by the interaction of the $\oInterpCmp[{\VxCxt[1] , \VxCxt[2]}]{\TxLet{\VxVar}{\VxCmp[2]}{\VxCmp[1]}}$ case with the exchange rule.
Renamings for both values and computations are preserved by Lem.~\ref{lemma:interp-sub}.
The $\oReturn$ functor is preserved by definition.
For uniqueness, let $\VoFunctor : \left( \xVALSMB , \xCMPSMB , \TxRet{} , \xFuncSMB , \xProcSMB \right) \longrightarrow \left( \oVal , \oCmp , \oReturn , \oFuncAsgSMB, \oProcAsgSMB \right)$ be a Freyd operad functor.
Given any $\VxVal \in \xVAL{n}$ or $\VxCmp \in \xCMP{n}$, mutual induction on $\VxVal$ and $\VxCmp$ shows that $\VoFunctor^{\oVal}(\VxVal)$ and $\VoFunctor^{\oCmp}(\VxCmp)$ are uniquely determined by $\VxVal$ and $\VxCmp$, since $\VoFunctor$ preserves the Freyd operad structure.
\end{proof}

If an equation is valid, it is satisfied in particular by the term model. But equality in the term model is, by definition, provable equality in the equational theory, which obtains completeness.

\begin{corollary}[Completeness]\label{thm:completeness}
Given well-formed terms $\VxTrm[1] , \VxTrm[2]$ in $\VxCxt$:
$\JxValid{\VxCxt}{\VxTrm[1]}{\VxTrm[2]} \; \Longrightarrow \; \JxEq{\VxCxt}{\VxTrm[1]}{\VxTrm[2]} $.
\end{corollary}

Now we finally get to the categorical semantics of $\xCLC$.
But rather than developing them independently, they arise immediately from the adjunction, and can equivalently be given directly in a Freyd PROP.
 Every weakly closed Freyd PROP equipped with interpretations of the symbols of $\xSign$ induces a $\xCLC$-structure by restriction to morphisms with codomain $1$. Therefore the interpretation, soundness, and completeness results above apply verbatim in the PROP setting. Combining this observation with the adjunction $F \dashv U$ and the initiality of the term model yields the corresponding PROP-level universal property.%

\section{Examples of weakly closed Freyd operads}\label{sec:examples}
In this section, we briefly illustrate how the semantic notion introduced above arises in several familiar settings. 
The first example supports our realizability motivation; the others demonstrate the structure in standard categorical models.

\vspace{0.2em}
\noindent
\textbf{Monadic Combinatory Algebras.} 
Monadic combinatory algebras (MCAs) were introduced in~\cite{cohen2025partial} as a generalization of partial combinatory algebras (PCAs), an abstraction of models for the untyped $\lambda$-calculus commonly used in realizability, to the effectful setting. They provide a direct source of weakly closed Freyd operads for $\xCLC$.
Let $(M,\mu,\eta)$ be a monad on $\mathbf{Set}$, and let $(\mca,\oApp)$ be a monadic applicative structure, where $\oApp : \mca^{2} \to M(\mca)$ is a Kleisli application operation.
Following~\cite{cohen2025partial}, the language of this structure consists of \emph{application expressions} over $n$ variables, built from variables, constants $c\in\mca$, and the binary operation $\bullet$; each such expression $e$ denotes, by the usual monadic interpretation, a Kleisli map $\mca^{n}\to M(\mca)$, and the maps so denoted are the \emph{$\mca$-monomials}. An $\mca$-monomial of the form $\eta \circ f$ for a definable $f : \mca^{n}\to\mca$ is a \emph{value monomial}.
The Freyd operad determined by $\mca$ is the operad of monomials: $\oV{n}$ is the set of value monomials $\mca^{n}\to\mca$, $\oC{n}$ the set of $\mca$-monomials $\mca^{n}\to M(\mca)$, and $\oReturn(f):=\eta\circ f$, with substitution and renaming given by substitution into and renaming of the variables of an expression.
A computation $m\in M(\mca)$ is \emph{computable} if it is of the form $\eta(c)$ for some $c\in\mca$. An $\mca$-monomial $f:\mca^{n+1}\to M(\mca)$ is \emph{computable} if there exists a function $\oAbs[n]{f}:\mca^{n}\to\mca$ such that $\eta\circ \oAbs[n]{f}$ is computable and, for every $c\in\mca$, application of $\oAbs[n]{f}$ to $c$ computes $f(-,c)$.
Thus, abstraction and application provide exactly the weak closure data for the Freyd operad of $\mca$-monomials. Then the induced Freyd operad is weakly closed precisely when all $\mca$-monomials are computable (the MCA analogue of combinatory completeness). This bridges our semantics with realizability-style models.

\vspace{0.2em}
\noindent
\textbf{Applicative Systems in Cartesian Restriction Categories.}
Another example comes from partial combinatory algebras internal to cartesian restriction categories~\cite{cockett2008introduction,cockett2010categories}. As observed in~\cite{cohen2025partial}, every cartesian restriction category carries a Freyd categorical structure determined by its restriction.
Let $\mca$ be an applicative system in a cartesian restriction category, with application morphism $\oApp : \mca^{2}\to \mca$. 
A morphism $f : \mca^{n} \rightarrow \mca$ is \emph{computable} when there is a total point $c_{f} : \boldsymbol{1} \rightarrow \mca$ such that $\oApp \circ \left( c_{f} \times \oId[\boldsymbol{\mca^{n}}]\right) = f$ and all partial applications are total.
One then obtains a Freyd operad $(\oVal,\oCmp,\oReturn)$ by taking $\oC{n}$ to be the morphisms $\mca^{n}\to\mca$, $\oV{n}$ to be the total morphisms $\mca^{n}\to\mca$, and $\oReturn$ to be the inclusion of total maps into all maps. In this case $\oCmp$ is in fact an operad, not merely a preoperad.
In~\cite{cockett2010categories}, combinatory completeness is formulated using a category of polynomials over $\mca$. Restricting attention to monomials suggests a closer match with the present framework: one expects a Freyd operad $(\oVal[\mca],\oCmp[\mca],\oReturn[\mca])$ whose computation part consists of monomials over $\mca$ and whose value part consists of the total ones, with weak closure corresponding to combinatory completeness. This comparison is especially suggestive because the Freyd category of substitutions generated by such an operad should play the role of the polynomial category.
We do not pursue this equivalence here, but it indicates that the substitutional viewpoint of Section~\ref{sec:freyd} is closely related to the restriction categorical account of computability.

\vspace{0.2em}
\noindent
\textbf{Reflexive Objects in a Cartesian Closed Category.}
Our final example recovers the usual higher-order semantics of reflexive objects~\cite{moggi1988computational}. 
Let $\mathcal{C}$ be a cartesian closed category equipped with a strong monad $M$, and let $\mca$ be an object equipped with maps
$e : M(\mca)^{\mca} \to \mca$
 and 
$d : \mca \to M(\mca)^{\mca}$.
If $e$ and $d$ form an isomorphism, then $\mca$ is a reflexive object in the usual sense; however,  if one assumes only $d \circ e = \oId$, one obtains the weaker intensional structure corresponding to $\beta$ without $\eta$.
From such data one obtains a weakly closed Freyd operad $(\oVal,\oCmp,\oReturn)$ by taking $\oV{n} := \mathcal{C}(\mca^{n},\mca)$, $\oC{n}$ to be the Kleisli maps $\mca^{n}\to M(\mca)$, and $\oReturn(f):=\eta\circ f$. Application is induced by the uncurrying of $d$, namely $\oApp := \oUncurry{d} : \mca^{2}\to M(\mca)$, and abstraction is induced by $e$: for $f:\mca^{n+1}\to M(\mca)$, define $\oAbs[n]{f}:= e\circ \oCurry{f} : \mca^{n}\to \mca$ (where $\oCurry{f}$ is the currying of $f$).
The weak closure axioms are then exactly the categorical forms of $\beta$, substitution compatibility, and naturality of abstraction.

\section{Conclusion}

We introduced Freyd operads as a substitutional structure for effectful Call-by-Value computation, where computation substitution is sequential. From any Freyd operad we constructed a Freyd PROP, proved representability, and established the expected adjunction.
This yields an initial semantic model for untyped computational $\lambda$-calculus with procedures and higher-order functions, and hence a sound and complete categorical account of its equational theory. 
More broadly, this shows that substitution can serve as an organizing principle for the semantics of effectful higher-order computation, opening a route toward syntactically faithful categorical frameworks with further semantic and logical applications.
Although this work focuses on the untyped setting, the extension to typed systems seems straightforward. The untyped case was prioritized primarily due to the specific difficulties it introduces for achieving completeness.

There are several natural directions for future work. First, we plan to relate the present framework more directly to realizability and to evaluation logic~\cite{pitts1991evaluation}, where the sequential character of computation is also fundamental. 
Second, it would be interesting to understand how the present Call-by-Value account interacts with Call-by-Push-Value~\cite{levy1999call}, and whether Freyd operads admit a corresponding reformulation in that setting. 
Third, we shall extend this work to a linear variant of computational $\lambda$-calculus using non-cartesian categories of values, and relate it to linear realizability~\cite{hoshino2007linear,tomita2021realizability}.
Finally, we plan to lift the construction to a more genuinely metacategorical level, for example by formulating it using internal categories in the sense of~\cite{hermida2000representable}, which may lead to a more flexible and conceptual account of substitutional semantics beyond the one-object case.


\bibliographystyle{entics}
\bibliography{ref}

\clearpage
\appendix

\section{Appendix}\label{sec:app}

\subsection{Category of Substitutions}

\begin{proof}\textbf{Prop.~\ref{prop:sub-constructions}}

\begin{enumerate}

\item\label{prop:sub-constructions:A}
Morphisms $m\to n$ are equivalence classes of pre-substitutions $\oPSUB[\oCmp]{m}{n}$ modulo
the symmetric rules of Fig.~\ref{fig:sub-rules}.
Because the relation is a congruence (rules apply in context), concatenation of lists descends to a well-defined,
associative composition with unit $\emptyset$.
Whiskering by $k$ on the left/right is defined componentwise on steps (single substitutions and renamings)
and hence on words.
Concretely, for single substitutions we define 
$    n \oSL \oSub{ m_{1} \oSL \VoCmp \oSR m_{2} } = \oSub{ n + m_{1} \oSL \VoCmp \oSR m_{2} } $ and $
    \oSub{ m_{1} \oSL \VoCmp \oSR m_{2} } \oSR n = \oSub{ m_{1} \oSL \VoCmp \oSR m_{2} + n }$, 
and for renamings define 
$    n \oSL \VoRen = \oId[n] + \VoRen $ and $
    \VoRen \oSR n = \VoRen + \oId[n]$.
Each generating rule of Fig.~\ref{fig:sub-rules} is stable under whiskering, so the operation is well-defined on the quotient.
Finally, take the symmetry $\sigma_{m,n}$ to be the singleton renaming word $[\sigma_{m,n}]$. 
Its naturality
follows from the renaming/substitution interaction \textbf{(N)} and the symmetry rules \textbf{(S$_\ell$)},\textbf{(S$_r$)}.

\item
If $\oVal$ is a cartesian operad, then $\oCSUBSMB[\oVal]$ is a pre-PROP as in~(\ref{prop:sub-constructions:A}).
The additional rule \textbf{(CEN)} of Fig.~\ref{fig:sub-rules}, with $\mathcal{P}(n) = \oV{n}$, is exactly the centrality
schema for adjacent substitution steps coming from $\oVal$.
Together with the commutation of a substitution step past a renaming step (an instance of \textbf{(N)} with identity blocks, using \textbf{(R)}), and since the relation is a congruence, it implies that
all morphisms in $\oCSUBSMB[\oVal]$ are central, hence $\oCSUBSMB[\oVal]$ is a PROP.
Define discard and copy by the singleton renaming words $[!_n]$ and $[\Delta_n]$.
The comonoid laws and the coherence conditions of Def.~\ref{def:cart-PROP} are inherited from $\oRENSMB^{op}$ (via composition of renaming steps).
The rules \textbf{(D)} and \textbf{(C)} give the naturality of $!_{1}$ and $\Delta_{1}$ against a single substitution step. The naturality required by Def.~\ref{def:cart-PROP}, namely that of $!_{n}$ and $\Delta_{n}$ against an arbitrary word $\Gamma : m \to n$, is obtained from these by induction on the length of $\Gamma$.

\item
Functoriality of $\oReturn^{\oSUBSMB}$ is due to the functoriality of $\oReturn$, the fact that it preserves renamings and substitution steps, and the last clause of Def.~\ref{def:op-sub}, which supplies in $\oSUBSMB[\oCmp]$ the images of the cartesian rules of $\oCSUBSMB[\oVal]$.
Centrality of the image of $\oReturn^{\oSUBSMB}$ in $\oSUBSMB[\oCmp]$ follows from rule \textbf{(CEN)} for $\mathcal{P}$ the image of $\oReturn$, which is available by Def.~\ref{def:op-sub} and sound by condition (i) of Def.~\ref{def:freyd-operad}, together with the commutation of substitution and renaming steps as in~(ii).

\item
Define $\VoFunctor^{\oCSUBSMB[\oVal]} : \oCSUBSMB[{\oVal[1]}] \longrightarrow \oCSUBSMB[{\oVal[2]}]$ and $\VoFunctor^{\oSUBSMB[\oCmp]} : \oSUBSMB[{\oCmp[1]}] \longrightarrow \oSUBSMB[{\oCmp[2]}]$ by recursion, using (resp.) $\VoFunctor^{\oVal}$ and $\VoFunctor^{\oCmp}$, the same way as in Def.~\ref{def:freyd-sub}.
Functoriality of (resp.) $\VoFunctor^{\oVal}$ and $\VoFunctor^{\oCmp}$ ensures that the equivalence relations are preserved.
Given a morphism $\VoVal \in \oCSUB[\oVal]{m}{n}$,  by induction on $\VoVal$, we obtain that $
\VoFunctor^{\oSUBSMB[\oCmp]}\!\bigl(\oReturn[1]^{\oSUBSMB}(\VoVal)\bigr) \equiv \oReturn[2]^{\oSUBSMB}\!\bigl(\VoFunctor^{\oCSUBSMB[\oVal]}(\VoVal)\bigr)
$, using the fact that all of the functors involved preserve renamings, and on morphisms in $\oVal$ we have $\VoFunctor^{\oCmp} \oReturn[1] \equiv \oReturn[2] \VoFunctor^{\oVal}$, ensuring the same for single substitutions.

\end{enumerate}
\end{proof}

\begin{proof}\textbf{Thm.~\ref{thm:pre-sub-equiv}}

By induction on the proof of $\VoSub[1] \cong \VoSub[2]$. Each case is proved through the following correspondences between the case and the appropriate rule in \Cref{fig:sub-rules}:
\begin{enumerate}
    \item Identities are discardable. The substitution case corresponds to $\textbf{U}_1$ in \Cref{fig:sub-rules}. The renaming case corresponds to $\textbf{U}_2$.

    \item Compatible substitutions are mergeable. Corresponds to $\textbf{A}$.

    \item Renamings are mergeable. Corresponds $\textbf{R}$.

    \item Compatible renamings and substitutions swap. Corresponds to $\textbf{N}$.

    \item Symmetries commute. Corresponds to $\textbf{S}_\ell$ and $\textbf{S}_r$.
    
    \item All substitutions are central. Corresponds to $\textbf{CEN}$.

     \item Discarding and copying are natural. $\textbf{D}$ and $\textbf{C}$.
\end{enumerate}
\end{proof}
\enlargethispage{-\baselineskip}

\subsection{Proof of Theorem~\ref{thm:preop-adj}}

\begin{lemma}\label{lem:forget-functor}
$\oForget : \oFREYDPROP \to \oFreydOp$ of Def.~\ref{def:forget} defines a functor.
\end{lemma}
\begin{proof}
We first show that  $\oForget$ is well defined on objects. 
The operadic identity and substitution laws in $\oVal^{\oForget}$ and $\oCmp^{\oForget}$ follow directly from the corresponding
identity and associativity laws in $\oVal$ and $\oCmp$ together with functoriality of whiskering.
The $\oRENSMB$-cartesian renaming action in $\oVal^{\oForget}$ is given by postcomposition
$f[r]=f\circ r$, and is functorial because composition in $\oVal$ is associative and unital.
For computations, renamings act by $f[r]=f\circ \oReturn(r)$, which is functorial since $\oReturn$ is a functor.
The Freyd compatibility between $\oReturn^{\oForget}$ and renaming actions is immediate from the definitions.

It remains to check that $(\oVal^{\oForget},\oCmp^{\oForget},\oReturn^{\oForget})$ satisfies the three conditions of Def.~\ref{def:freyd-operad}.
Centrality of $\oReturn^{\oForget}$-images in $\oCmp^{\oForget}$ is inherited from the centrality of $\oReturn$-images in the pre-PROP $\oCmp$, since operadic substitution in $\oCmp^{\oForget}$ is whiskering followed by composition.
For discardability and copyability, let $\VoVal \in \oHom{\oVal}{n}{1}$ and $\VoCmpI \in \oHom{\oCmp}{m_{1}+m_{2}}{1}$, respectively $\VoCmpI \in \oHom{\oCmp}{m_{1}+2+m_{2}}{1}$. Unfolding the definitions, the two required identities are obtained by precomposing $\VoCmpI$ with the naturality squares for $\oDiscard[]$ and $\oCopy[]$ in the cartesian PROP $\oVal$, transported into $\oCmp$ by the pre-PROP functor $\oReturn$:
$$
\oReturn\!\left(\oId[m_{1}]+\oDiscard[1]+\oId[m_{2}]\right) \circ \left(m_{1} \oSL \oReturn \VoVal \oSR m_{2}\right)
= \oReturn\!\left( \left(\oId[m_{1}]+\oDiscard[1]+\oId[m_{2}]\right)\circ\left(m_{1} \oSL \VoVal \oSR m_{2}\right) \right)
= \oReturn\!\left(\oId[m_{1}]+\oDiscard[n]+\oId[m_{2}]\right),
$$
and likewise for $\oCopy[]$.

To see that $\oForget$ is functorial, recall that on morphisms it is restriction to codomain $1$. Restrictions preserve identities and composition.
Moreover, the Freyd square for a Freyd PROP functor,
$\VoFunctor^{\oCmp}\oReturn[1]\equiv \oReturn[2]\VoFunctor^{\oVal}$, restricts to codomain $1$ to give the corresponding Freyd
square in $\oFreydOp$.  
\end{proof}

\begin{lemma}\label{lem:free-functor}
    $\oFree : \oFreydOp \to \oFREYDPROP$ from Def.~\ref{def:free} defines a functor.
\end{lemma}
\begin{proof}
Well-definedness on objects is given by Def.~\ref{def:freyd-sub}. Well-definedness on morphisms is given by Prop.~\ref{prop:sub-constructions}.
To prove preservation of identities, consider the identity Freyd operad functor over an arbitrary Freyd operad, and then for each of its two components, given an arbitrary substitution in the corresponding category, show by induction on the substitution that the image of $\oFree$ on the appropriate component preserves the substitution, which follows from the functoriality of the component.
Proving preservation of composition works the same way.
\end{proof}

To prove the adjunction $\oFree \dashv \oForget$, we specify two operations on functors, one of currying, turning a Freyd PROP functor from a Freyd PROP of substitutions into a Freyd operad functor from the underlying Freyd operad, and another one of uncurrying, in the opposite direction.
Then, we show that they form an isomorphism.

\begin{defn}[Currying]\label{def:curry}
Let $\VoOperad\in \oFreydOp$ and $\VoProp\in \oFREYDPROP$.
Given a Freyd PROP functor
$
(\VoFunctor^{\oVal},\VoFunctor^{\oCmp}) : \oFree(\VoOperad) \to \VoProp
$,
we define a Freyd operad functor
$
\oCurry{(\VoFunctor^{\oVal},\VoFunctor^{\oCmp})}
: \VoOperad \to \oForget(\VoProp)
$
by setting, for each $f\in \oV{n}$ and $g\in \oC{n}$,
$$
\oCurry{(\VoFunctor^{\oVal})}(f) \;:=\; \VoFunctor^{\oVal}\bigl(\lPar{\oSub{0 \oSL f \oSR 0}}\bigr),
\qquad
\oCurry{(\VoFunctor^{\oCmp})}(g) \;:=\; \VoFunctor^{\oCmp}\bigl(\lPar{\oSub{0 \oSL g \oSR 0}}\bigr).
$$
\end{defn}

\begin{lemma}\label{lem:curry-functor}
$
\oCurry{(\VoFunctor^{\oVal},\VoFunctor^{\oCmp})} = (\oCurry{\VoFunctor^{\oVal}},\oCurry{\VoFunctor^{\oCmp}})
: \VoOperad \to \oForget(\VoProp)
$ defines a Freyd operad functor.
\end{lemma}
\begin{proof}
Functoriality follows from functoriality of $(\VoFunctor^{\oVal},\VoFunctor^{\oCmp})$ and the definition of composition in
$\oFree(\VoOperad)$ via concatenation of substitution lists.
For Freyd compatibility, for any $v\in \oV{n}$ we compute:
$$
(\oForget\oReturn[2])\,\oCurry{(\VoFunctor^{\oVal})}\,v
= \oReturn[2]\bigl(\VoFunctor^{\oVal}(\lPar{\oSub{0 \oSL v \oSR 0}})\bigr)
= \VoFunctor^{\oCmp}\bigl(\oReturn[1]^{\oSUBSMB}(\lPar{\oSub{0 \oSL v \oSR 0}})\bigr)
= \VoFunctor^{\oCmp}\bigl(\lPar{\oSub{0 \oSL \oReturn[1]v \oSR 0}}\bigr)
= \oCurry{\VoFunctor^{\oCmp}}\,\oReturn[1]\,v.
$$
\end{proof}

\begin{defn}[Uncurrying]\label{def:uncurry}
Let $\VoOperad=(\oVal[1],\oCmp[1],\oReturn[1])\in \oFreydOp$ and let
$\VoProp=(\oVal[2],\oCmp[2],\oReturn[2])\in \oFREYDPROP$.
Given a Freyd operad functor
$
(\VoFunctor^{\oVal},\VoFunctor^{\oCmp}) : \VoOperad \to \oForget(\VoProp)
$, 
we define functors
$$
\oUncurry{\VoFunctor^{\oVal}} : \oCSUBSMB[{\oVal[1]}] \to \oVal[2] \qquad
\qquad
\oUncurry{\VoFunctor^{\oCmp}} : \oSUBSMB[{\oCmp[1]}] \to \oCmp[2]
$$
by recursion on substitution lists. 
For $\VoSub\in \oCSUB[{{\oVal[1]}}]{m}{n}$ we set
$$
\begin{array}{lcl}
\oUncurry{\VoFunctor^{\oVal}}\,\lEmpty &:=& \oId,\\
\oUncurry{\VoFunctor^{\oVal}}\,\lPar{\oSub{m_1 \oSL f \oSR m_2},\,\Gamma}
&:=& (m_1 \oSL \VoFunctor^{\oVal}(f) \oSR m_2)\circ \oUncurry{\VoFunctor^{\oVal}}\,\Gamma,\\
\oUncurry{\VoFunctor^{\oVal}}\,\lPar{r,\,\Gamma}
&:=& r \circ \oUncurry{\VoFunctor^{\oVal}}\,\Gamma,
\end{array}
$$
where $r$ is interpreted in $\oVal[2]$ via Lemma~\ref{lemma:ren-gen}.
The definition of $\oUncurry{\VoFunctor^{\oCmp}}$ is identical, except that renaming steps are interpreted as
$\oReturn[2](r)\circ (-)$.
\end{defn}

\begin{lemma}\label{lem:uncurry-functorial}
$\oUncurry{\VoFunctor^{\oVal}}$ and $\oUncurry{\VoFunctor^{\oCmp}}$ preserve identities and composition.
\end{lemma}
\begin{proof}
Identity is preserved by definition.
For composition,
given two compatible substitutions $\VoSub[1] , \VoSub[2]$, we analyze $\oUncurry{\VoFunctor^{\oVal}} \left(\VoSub[1] , \VoSub[2]\right)$ by induction on $\VoSub[1]$:
\[
\begin{array}{lcl}
    \oUncurry{\VoFunctor^{\oVal}} \left(\lEmpty \lApp \VoSub[2]\right) & = & \oUncurry{\VoFunctor^{\oVal}} \VoSub[2]\\
    &=&\lEmpty , \oUncurry{\VoFunctor^{\oVal}} \VoSub[2]\\
    &=& \oUncurry{\VoFunctor^{\oVal}} \lEmpty , \oUncurry{\VoFunctor^{\oVal}} \VoSub[2] \\
    \\
    \oUncurry{\VoFunctor^{\oVal}} \left( \oSub{ m_{1} \oSL \VoVal \oSR m_{2} } , \Gamma \lApp \VoSub[2]\right)
    & = & \oUncurry{\VoFunctor^{\oVal}} \lPar{ \oSub{ m_{1} \oSL \VoVal \oSR m_{2} } , \Gamma \lApp \VoSub[2] } \\
    & = & \left(m_{1} \oSL \VoFunctor^{\oVal}  \VoVal \oSR m_{2}\right) ,  \oUncurry{\VoFunctor^{\oVal}} \left(\Gamma \lApp \VoSub[2]\right) \\
    & = & \left(m_{1} \oSL \VoFunctor^{\oVal}  \VoVal \oSR m_{2}\right) ,  \oUncurry{\VoFunctor^{\oVal}} \Gamma , \oUncurry{\VoFunctor^{\oVal}} \VoSub[2] \\
    & = &
    \oUncurry{\VoFunctor^{\oVal}} \left( \oSub{ m_{1} \oSL \VoVal \oSR m_{2} } , \Gamma \right) , \oUncurry{\VoFunctor^{\oVal}} \VoSub[2]\\
    \\
    \oUncurry{\VoFunctor^{\oVal}} \left( \oRen{\VoRen} , \Gamma \lApp \VoSub[2]\right) & = & \oUncurry{\VoFunctor^{\oVal}} \lPar{ \oRen{\VoRen} , \Gamma \lApp \VoSub[2] }\\
    &=& \oRen{\VoRen} , \oUncurry{\VoFunctor^{\oVal}} \left(\Gamma \lApp \VoSub[2]\right)\\
    &=& \oRen{\VoRen} , \oUncurry{\VoFunctor^{\oVal}} \Gamma , \oUncurry{\VoFunctor^{\oVal}} \VoSub[2]\\
    &=& \oUncurry{\VoFunctor^{\oVal}} \lPar{ \oRen{\VoRen} , \Gamma } , \oUncurry{\VoFunctor^{\oVal}} \VoSub[2]
\end{array}
\]
So we get $\oUncurry{\VoFunctor^{\oVal}}\lPar{\VoSub[1] , \VoSub[2]} = \oUncurry{\VoFunctor^{\oVal}}\lPar{\VoSub[1]} , \oUncurry{\VoFunctor^{\oVal}}\lPar{\VoSub[2]}$ for all compatible substitutions $\VoSub[1] , \VoSub[2]$.
The treatment of $\VoFunctor^{\oCmp}$ is very similar, but there, in the case of renamings, we get: $\oUncurry{\VoFunctor^{\oCmp}} \lPar{ \oRen{\VoRen} , \Gamma } = \oReturn \oRen{\VoRen} ,  \oUncurry{\VoFunctor^{\oCmp}} \Gamma$, in contrast to $\oUncurry{\VoFunctor^{\oVal}} \lPar{ \oRen{\VoRen} , \Gamma } = \oRen{\VoRen} ,  \oUncurry{\VoFunctor^{\oVal}} \Gamma$. However, the rest of the case works exactly the same way.
 
\end{proof}

Next we show that uncurrying respects the quotient.

\begin{lemma}\label{lem:uncurry-quot}
$\oUncurry{\VoFunctor^{\oVal}}$ and $\oUncurry{\VoFunctor^{\oCmp}}$ respect the equivalence relation defining
$\oCSUBSMB[{\oVal[1]}]$ and $\oSUBSMB[{\oCmp[1]}]$.
\end{lemma}
\begin{proof}
First we show by induction that $\oUncurry{\VoFunctor^{\oVal}}$ preserves the equivalence relation, using the functoriality of $\VoFunctor^{\oVal}$, its preservation of renamings, and the properties of $\oVal[2]$ as a cartesian PROP.
We do that for each case of Fig.~\ref{fig:sub-rules} separately:
\begin{enumerate}
    \item \textbf{(U$_1$)}:
    \[
\begin{array}{lcl}
    \oUncurry{\VoFunctor^{\oVal}} \lPar{ \Gamma , \oSub{ m_{1} \oSL \oId \oSR m_{2} } , \Gamma' } & = & \oUncurry{\VoFunctor^{\oVal}} \Gamma \circ \paren{ m_{1} \oSL \VoFunctor^{\oVal} \oId \oSR m_{2} }  \circ \oUncurry{\VoFunctor^{\oVal}} \Gamma' \\
    & = & \oUncurry{\VoFunctor^{\oVal}} \Gamma \circ \paren{ m_{1} \oSL \oId \oSR m_{2} } \circ \oUncurry{\VoFunctor^{\oVal}} \Gamma'\\
    & = & \oUncurry{\VoFunctor^{\oVal}} \Gamma \circ \oUncurry{\VoFunctor^{\oVal}} \Gamma'\\
    & = & \oUncurry{\VoFunctor^{\oVal}} \lPar{\Gamma , \Gamma'}\\
\end{array}
\]

    \item \textbf{(A)}:
    \[
\begin{array}{lcl}
    && \oUncurry{\VoFunctor^{\oVal}} \lPar{ \Gamma , \oSub{m_{1} \oSL \VoCmpII \oSR m_{2}} , \oSub{m_{1} + n_{1} \oSL \VoCmpIII \oSR n_{2} + m_{2}} , \Gamma'}\\
    &=& \oUncurry{\VoFunctor^{\oVal}} \Gamma \circ \paren{ m_{1} \oSL \VoFunctor^{\oVal} \VoCmpII \oSR m_{2} } \circ \paren{ m_{1} + n_{1} \oSL \VoFunctor^{\oVal} \VoCmpIII \oSR n_{2} + m_{2} } \circ \oUncurry{\VoFunctor^{\oVal}}\Gamma'\\
    &=& \oUncurry{\VoFunctor^{\oVal}} \Gamma \circ \paren{ m_{1} \oSL \VoFunctor^{\oVal} \VoCmpII \oSR m_{2} } \circ \paren{ m_{1} \oSL \paren{ n_{1} \oSL \VoFunctor^{\oVal} \VoCmpIII \oSR n_{2} } \oSR m_{2} } \circ \oUncurry{\VoFunctor^{\oVal}}\Gamma'\\
    &=& \oUncurry{\VoFunctor^{\oVal}} \Gamma \circ \paren{ m_{1} \oSL \left( \VoFunctor^{\oVal} \VoCmpII \circ \paren{n_{1} \oSL \VoFunctor^{\oVal} \VoCmpIII  \oSR n_{2}}\right) \oSR m_{2} } \circ \oUncurry{\VoFunctor^{\oVal}}\Gamma'\\
    &=& \oUncurry{\VoFunctor^{\oVal}}\Gamma \circ \left( m_{1} \oSL \VoFunctor^{\oVal} \left(\oSubAp{ \VoCmpII }{ n_{1} \oSL \VoCmpIII \oSR n_{2} }\right) \oSR m_{2} \right) \circ \oUncurry{\VoFunctor^{\oVal}}\Gamma'\\
    &&\oUncurry{\VoFunctor^{\oVal}} \lPar{ \Gamma , \oSub{ m_{1} \oSL \oSubAp{\VoCmpII}{n_{1} \oSL \VoCmpIII \oSR n_{2}} \oSR m_{2} } , \Gamma' }\\
\end{array}
\]

\item \textbf{(U$_2$)}, \textbf{(R)}, \textbf{(N)}, \textbf{(S$_\ell$)}, \textbf{(S$_r$)}, \textbf{(CEN)}, \textbf{(D)}, \textbf{(C)}: Follow from Prop.~\ref{prop:sub-constructions}, as properties of a cartesian PROP.

\end{enumerate}

The treatment of $\oUncurry{\VoFunctor^{\oCmp}}$ is very similar, with the following key differences:
\begin{enumerate}
    \item For renamings, we get: $\oUncurry{\VoFunctor^{\oCmp}} \lPar{ \oRen{\VoRen} , \Gamma } = \oReturn[2] \VoRen \circ  \oUncurry{\VoFunctor^{\oCmp}} \Gamma$, in contrast to $\oUncurry{\VoFunctor^{\oVal}} \lPar{ \oRen{\VoRen} , \Gamma } = \oRen{\VoRen} \circ  \oUncurry{\VoFunctor^{\oVal}} \Gamma$. However, since $\oReturn[2] \VoRen$ is always central, the treatment is almost identical.
    \item By Def.~\ref{def:op-sub}, the rules \textbf{(CEN)}, \textbf{(D)} and \textbf{(C)} are imposed on $\oSUBSMB[{\oCmp[1]}]$ only for steps $\oSub{m_{1} \oSL \oReturn[1] \VoVal \oSR m_{2}}$ whose operation lies in the image of $\oReturn[1]$. On such a step $\oUncurry{\VoFunctor^{\oCmp}}$ produces $m_{1} \oSL \oReturn[2]\left(\VoFunctor^{\oVal} \VoVal\right) \oSR m_{2}$, by the Freyd square for $(\VoFunctor^{\oVal},\VoFunctor^{\oCmp})$. The three resulting identities then hold in $\oCmp[2]$ because $\VoProp$ is a Freyd PROP: centrality of $\oReturn[2]$-images gives \textbf{(CEN)}, while applying the pre-PROP functor $\oReturn[2]$ to the naturality of $\oDiscard[]$ and $\oCopy[]$ in $\oVal[2]$ gives \textbf{(D)} and \textbf{(C)}. These are exactly the discardability and copyability conditions of Def.~\ref{def:freyd-operad}, holding in $\oForget(\VoProp)$.

\end{enumerate}
\end{proof}

\begin{lemma}\label{lem:uncurry-freyd}
$$
\oUncurry{\VoFunctor^{\oCmp}}\,\oReturn[1]^{\oSUBSMB} \;\equiv\; \oReturn[2]\,\oUncurry{\VoFunctor^{\oVal}}.
$$
\end{lemma}
\begin{proof}
By induction on substitution lists, unfolding the definition of $\oReturn^{\oSUBSMB}$ and the recursive clauses of
$\oUncurry{\VoFunctor^{\oVal}}$ and $\oUncurry{\VoFunctor^{\oCmp}}$.
\end{proof}

\begin{corollary}
\label{cor:uncurry-functor}
$\oUncurry{(\VoFunctor^{\oVal},\VoFunctor^{\oCmp})} = (\oUncurry{\VoFunctor^{\oVal}},\oUncurry{\VoFunctor^{\oCmp}}) : \oFree(\VoOperad) \to \VoProp
$ defines a Freyd PROP functor.
\end{corollary}

Next, we establish that Curry and Uncurry are inverses. 

\begin{lemma}
\label{lem:curry-uncurry}
The following hold:
\begin{itemize}
    \item For a Freyd PROP functor $(\VoFunctor^{\oVal},\VoFunctor^{\oCmp}) : \oFree(\VoOperad)\to \VoProp$,
$
\oUncurry{\oCurry{\VoFunctor^{\oVal}}} \equiv \VoFunctor^{\oVal}
$ and $
\oUncurry{\oCurry{\VoFunctor^{\oCmp}}} \equiv \VoFunctor^{\oCmp}
$.
\item For a Freyd operad functor $(\VoFunctor^{\oVal},\VoFunctor^{\oCmp}) : \VoOperad \to \oForget(\VoProp)$ and any $f$,
$
\oCurry{\oUncurry{\VoFunctor^{\oVal}}}(f) = \VoFunctor^{\oVal}(f)
$ and $
\oCurry{\oUncurry{\VoFunctor^{\oCmp}}}(f) = \VoFunctor^{\oCmp}(f).
$
\end{itemize}
\end{lemma}
\begin{proof}
\begin{itemize}
    \item 
By induction on substitution lists, unfolding Def.~\ref{def:curry} and~\ref{def:uncurry} and using functoriality of $\VoFunctor^{\oVal}$ and
$\VoFunctor^{\oCmp}$.

For $\VoFunctor^{\oVal}$, we get by induction:
\begin{align*}
    \oUncurry{\oCurry{\VoFunctor^{\oVal}}} \lEmpty & = \oId = \VoFunctor^{\oVal} \lEmpty \displaybreak[1]\\[\baselineskip]
    \oUncurry{\oCurry{\VoFunctor^{\oVal}}} \lPar{ \oSub{ m_{1} \oSL \VoCmp \oSR m_{2} } , \Gamma } & = \left(m_{1} \oSL \oCurry{\VoFunctor^{\oVal}} \VoCmp \oSR m_{2}\right) \circ  \oUncurry{\oCurry{\VoFunctor^{\oVal}}} \Gamma\\
    &= \left(m_{1} \oSL \VoFunctor^{\oVal} \lPar{\oSub{ 0 \oSL \VoCmp \oSR 0 }} \oSR m_{2}\right) \circ  \VoFunctor^{\oVal} \Gamma\\
    &= \VoFunctor^{\oVal} \lPar{\oSub{ m_{1} \oSL \VoCmp \oSR m_{2} }} \circ  \VoFunctor^{\oVal} \Gamma\\
    &= \VoFunctor^{\oVal} \lPar{\oSub{ m_{1} \oSL \VoCmp \oSR m_{2} } , \Gamma} \displaybreak[1]\\[\baselineskip]
    \oUncurry{\oCurry{\VoFunctor^{\oVal}}} \lPar{ \VoRen , \Gamma } &= \VoRen \circ \oUncurry{\oCurry{\VoFunctor^{\oVal}}} \Gamma\\
    &= \VoRen \circ \VoFunctor^{\oVal} \Gamma\\
    &= \VoFunctor^{\oVal} \lPar{ \VoRen , \Gamma }
\end{align*}
The same calculation works with $\VoFunctor^{\oCmp}$, but with $\oReturn[2] \VoRen \circ$ instead of $\VoRen \circ$ wherever it appears.

\item Unfold the definitions: $\oCurry{\oUncurry{\VoFunctor^{\oVal}}}(f)=\oUncurry{\VoFunctor^{\oVal}}(\lPar{\oSub{0 \oSL f \oSR 0}})=\VoFunctor^{\oVal}(f)$,
and similarly for $\VoFunctor^{\oCmp}$.

\end{itemize}
\end{proof}

\begin{corollary}
\[ \oFree : \oFreydOp \longrightarrow \oFREYDPROP \]
and
\[ \oForget : \oFREYDPROP \longrightarrow \oFreydOp \]
form an adjunction
$\oFree \dashv \oForget$.
\end{corollary}

\subsection{Interpretation}

\begin{proof} \textbf{Lemma~\ref{lemma:interp-sub}}.
\\
By mutual induction on the derivations of $\JxWFV{\VxVar[1],\ldots,\VxVar[n]}{\VxVal}$ and $\JxWFC{\VxVar[1],\ldots,\VxVar[n]}{\VxCmp}$.
The variable, function, return, let, procedure, abstraction, and application cases follow directly from the recursive definition of interpretation together with the operadic substitution laws, functoriality of $\oReturn$, and the substitution-compatibility axiom of weak closure.
\end{proof}

\begin{proof}\textbf{Thm.~\ref{thm:free-model}}

The identity morphism in $\xVALSMB$ is the well-formed value $\JxWFV{\VxVar}{\VxVar}$.
Composition of $\VxVal[2]$ in $\VxVal[1] \in \xVAL{n_{1} + 1 + n_{2}}$, assuming $\JxWFV{\VxVarI[1] , \ldots  , \VxVarI[n_{1}] , \VxVarII , \VxVarIII[1] , \ldots , \VxVarIII[n_{2}] }{ \VxVal[1] }$, is given by:
\[ \oSubAp{\VxVal[1]}{n_{1} \oSL \VxVal[2] \oSR n_{2}} := \VxVal[1] \subst{ \assign{\VxVarII}{\VxVal[2]} } \]
Assuming $\JxWFV{\VxVar[1] , \ldots , \VxVar[m]}{ \VxVal }$, and $\VoRen \in \oREN{m}{n}$, semantic renaming with $\VoRen$ from context $\VxVarI[1] , \ldots , \VxVarI[m]$ to context $\VxVarII[1] , \ldots , \VxVarII[n]$ is given by syntactic renaming:
\[ \oRenAp{\VxVal}{\VoRen} := \VxVal \subst{ \assign{\VxVarI[1]}{\VxVarII[\VoRen\left(1\right)]} , \ldots , \assign{\VxVarI[m]}{\VxVarII[\VoRen\left(m\right)]} } \]
The fact that $\xVALSMB$ is a cartesian operad follows directly from the properties of substitution.

The identity morphism in $\xCMPSMB$ is the well-formed computation $\JxWFC{\VxVar}{ \TxRet{\VxVar}}$.
Composition of $\VxCmp[2] \in \xCMP{n}$ in $\VxCmp[1] \in \xCMP{n_{1} + 1 + n_{2}}$, assuming $\JxWFC{\VxVarI[1] , \ldots  , \VxVarI[n_{1}] , \VxVarII , \VxVarIII[1] , \ldots , \VxVarIII[n_{2}] }{ \VxCmp[1] }$,  is given by:
\[ \oSubAp{\VxCmp[1]}{n_{1} \oSL \VxCmp[2] \oSR n_{2}} := \TxLet{\VxVarII}{\VxCmp[2]}{\VxCmp[1]} \]
using the exchange rule to fit $\VxCmp[1]$ to $\JxWFC{\VxVarI[1] , \ldots  , \VxVarI[n_{1}] , \VxVarIII[1] , \ldots , \VxVarIII[n_{2}] , \VxVarII }{ \VxCmp[1] }$

Semantic renaming is given by syntactic renaming.
The first two rules of $\oRENSMB$-cartesian preoperad are immediate from the properties of substitution.
To prove the third rule, let $\VxCmp[1] \in \xCMP{n_{1} + 1 + n_{2}}$, $\VxCmp[2] \in \xCMP{n}$, $\VoRenI[i] \in \oREN{m_{i}}{n_{i}}$, and $\VoRenII \in \oREN{m}{n}$. We need to show:
$$ \oRenAp{\oSubAp{\VxCmp[1]}{ n_{1} \oSL \VxCmp[2]  \oSR n_{2} }}{ \VoRenI[1] + \VoRenII + \VoRenI[2]  } \equiv \oSubAp{ \oRenAp{\VxCmp[1]}{\VoRenI[1] + \oId[1] + \VoRenI[2]}  }{ m_{1} \oSL \oRenAp{\VxCmp[2]}{\VoRenII} \oSR m_{2} } $$
Assuming $\JxWFC{\VxVarI[1] , \ldots  , \VxVarI[n_{1}] , \VxVarII , \VxVarIII[1] , \ldots , \VxVarIII[n_{2}] }{ \VxCmp[1] }$ and $\JxWFC{\VxVarIV[1] , \ldots , \VxVarIV[n]}{\VxCmp[2]}$, it means:
$$\begin{array}{lcl}
    && \left(\TxLet{\VxVarII}{\VxCmp[2]}{\VxCmp[1]}\right) \subst{ \assign{\VxVarI[1]}{\VxVarI[{\VoRenI[1]\left(1\right)}]'} , \ldots , \assign{\VxVarI[n_{1}]}{\VxVarI[{\VoRenI[1]\left(n_{1}\right)}]'} , \assign{\VxVarIV[1]}{\VxVarIV[{\VoRenII\left(1\right)}]'} , \ldots , \assign{\VxVarIV[n]}{\VxVarIV[{\VoRenII\left(n\right)}]'} , \assign{\VxVarIII[1]}{\VxVarIII[{\VoRenI[2]\left(1\right)}]'} , \ldots , \assign{\VxVarIII[n_{2}]}{\VxVarIII[{\VoRenI[2]\left(n_{2}\right)}]'} } \\
    &\equiv
    &\TxLet{\VxVarII}{\VxCmp[1] \subst{\assign{\VxVarI[1]}{\VxVarI[{\VoRenI[1]\left(1\right)}]'} , \ldots , \assign{\VxVarI[n_{1}]}{\VxVarI[{\VoRenI[1]\left(n_{1}\right)}]'} , \assign{\VxVarIII[1]}{\VxVarIII[{\VoRenI[2]\left(1\right)}]'} , \ldots , \assign{\VxVarIII[n_{2}]}{\VxVarIII[{\VoRenI[2]\left(n_{2}\right)}]'}} }{ \VxCmp[2] \subst{ \assign{\VxVarIV[1]}{\VxVarIV[{\VoRenII\left(1\right)}]'} , \ldots , \assign{\VxVarIV[n]}{\VxVarIV[{\VoRenII\left(n\right)}]'} } }
\end{array}$$
Which follows directly from the properties of substitution.
The interaction of exchange and $\TxLet{-}{-}{-}$ ensures that $\xCMPSMB$ is a symmetric preoperad.

We define a preoperad functor between them $\TxRetSMB : \xVALSMB \longrightarrow \xCMPSMB$ by the obvious definition, taking each $\VxVal \in \xVAL{n}$ to $\TxRet{\VxVal} \in \xCMP{n}$.
The identity is preserved by $\TxRetSMB$ by definition. Composition is preserved by the lunit rule.
Preserving centrality means that $\TxLet{\VxVar[1]}{\TxRet{\VxVal[1]}}{\TxLet{\VxVar[2]}{\VxCmp[2]}{\VxCmp}} 
    \equiv
    \TxLet{\VxVar[2]}{\VxCmp[2]}{\TxLet{\VxVar[1]}{\TxRet{\VxVal[1]}}{\VxCmp}}$,
which is the result of lunit rule.
Weakening and contraction are preserved through the lunit rule.
Renamings are preserved by the definition of substitution.

The weakly closed structure is given by abstraction and application. $\oApp \in \xCMP{2}$ is the well-formed computation derived from the application of one variable on another $\JxWFC{\VxVar[1] , \VxVar[2]}{ \TxApp{\VxVar[1]}{\VxVar[2]} }$.
Given a computation $\VxCmp \in \xCMP{n+1}$, assuming $\JxWFC{\VxCxt , \VxVar}{\VxCmp}$, we define $\oAbs[n]{\VxCmp} := \TxAbs{\VxVar}{\VxCmp}$.
The first equation of the weakly closed Freyd operad is proved by applying the beta rule.
The second equation of the weakly closed Freyd operad is proved using the lunit rule and the properties of substitution.
The third equation, naturality of $\oAbs[n]{-}$ in the surrounding context, holds on the nose: renaming the free variables of $\TxAbs{\VxVar}{\VxCmp}$ is renaming those of $\VxCmp$ other than $\VxVar$, so $\TxAbs{\VxVar}{\left(\oRenAp{\VxCmp}{\VoRen+\oId[1]}\right)} = \oRenAp{\left(\TxAbs{\VxVar}{\VxCmp}\right)}{\VoRen}$ as terms, since $\VxVar$ is bound and is left untouched by $\VoRen+\oId[1]$.

There are no axioms governing $\oFuncAsgSMB$ and $\oProcAsgSMB$ besides well-formedness, which is trivial.
\end{proof}

\subsection{Interpretation in Freyd PROPs}

Many models are presented not only as Freyd operads, but as Freyd PROPs.
Recall that any Freyd PROP $\left(\oVal,\oCmp,\oReturn\right)$ induces a Freyd operad by restricting to morphisms with codomain $1$.
Concretely, we set $\oVal(n) := \oHom{\oVal}{n}{1}$ and $\oCmp(n) := \oHom{\oCmp}{n}{1}$, with substitution induced by composition and tensor in the PROP.
In particular, every Freyd PROP yields a $\xCLC$-structure in the sense of Def.~\ref{def:structure} once we specify interpretations for the function and procedure symbols.

To finish the discussion, we now briefly discuss the categorical semantics for $\xCLC$ directly in a Freyd PROP.
To that end, we define a categorical version of $\xCLC$-structure:
\begin{defn}[Categorical $\xCLC$-Structure]
Given a signature $\xSign = \left( \xFuncSMB , \xProcSMB \right)$, a \emph{categorical $\xCLC$-structure} over $\xSign$ is a Freyd PROP, equipped with a morphism $\oFuncAsg{f} \in \oHom{\oVal}{n}{1}$ for every function symbol $f \in \xFuncSymb{n}$, a morphism $\oProcAsg{p} \in \oHom{\oCmp}{n}{1}$ for every procedure symbol $p \in \xProcSymb{n}$, and a natural transformation $\oAbs[n]{-} : \oHom{\oCmp}{n+1}{1} \longrightarrow \oHom{\oVal}{n}{1}$ of \emph{abstraction}, and a morphism $\oApp \in \oHom{\oCmp}{2}{1}$ of \emph{application}, such that  the following holds for every $\VoCmp \in \oHom{\oCmp}{n+1}{1}$: $\oApp \circ \left( \oReturn \oAbs[n]{f} \oSR 1 \right) = f$.
\end{defn}

The interpretation of $\xCLC$ in a categorical $\xCLC$-structure is obtained by applying Def.~\ref{def:interp-freyd} to the induced $\xCLC$-structure on maps into $1$.

\begin{defn}[Morphism of Categorical Structures]
A morphism $\VoFunctor : \VoStruct[1] \to \VoStruct[2]$ of categorical structures is a Freyd PROP functor between the underlying Freyd PROPs that also preserves application, abstraction, and the symbol interpretations.
\end{defn}
The following lemma is immediate from the definition:
\begin{lemma}
Given a categorical structure $\VoStruct = \left( \oVal , \oCmp , \oReturn , \oFuncAsgSMB , \oProcAsgSMB \right)$, the forgetful functor yields a structure  $\oForget\VoStruct := \left( \oForget\oVal , \oForget\oCmp , \oForget\oReturn , \oForget \oFuncAsgSMB , \oForget \oProcAsgSMB \right)$ in the sense of \Cref{def:structure}.
\end{lemma}

This insight leads us naturally to the \emph{categorical} semantics of $\xCLC$: To interpret $\xCLC$ in a categorical structure $\VoStruct$, we simply use \Cref{def:interp-freyd} to interpret it in $\oForget \VoStruct$.

Combining \Cref{thm:preop-adj} with \Cref{thm:initiality}, we get the following:
\begin{corollary}\label{cor:prop-initiality}
$\left( \oCSUBSMB[\xVALSMB] , \oSUBSMB[\xCMPSMB] , \oReturn^{\oSUBSMB} \right) $ is the initial object in the category of weakly closed Freyd PROPs.
\end{corollary}

Note that the soundness and completeness results carry over to the semantics in a categorical structure.
This is a direct corollary of \Cref{cor:prop-initiality} and \Cref{thm:sub-presheaf} with (resp.) \Cref{thm:soundness} and \Cref{thm:completeness}.

\end{document}